\documentclass[11pt]{article}
\usepackage{amssymb}
\usepackage{amsfonts}
\usepackage{amsmath}
\usepackage[nohead]{geometry}
\usepackage[singlespacing]{setspace}
\usepackage[bottom]{footmisc}
\usepackage{indentfirst}
\usepackage{endnotes}
\usepackage{graphicx}%
\usepackage{rotating}
\usepackage{amsthm}
\usepackage{booktabs}
\usepackage{algorithm}
\usepackage{algpseudocode}
\usepackage{tikz}
\RequirePackage[OT1]{fontenc} \RequirePackage{amsthm,amsmath}
\RequirePackage{natbib}
\RequirePackage[colorlinks,citecolor=blue,urlcolor=blue]{hyperref}
\RequirePackage{hypernat} \makeatletter
\newcommand{\field}[1]{\mathbb{#1}}

\newcommand{\E}{\field{E}}

\theoremstyle{example} \theoremstyle{remark} \theoremstyle{lemma}
\theoremstyle{definition} \theoremstyle{corol}
\theoremstyle{proposition} \theoremstyle{condition}
\theoremstyle{assumption}

\newtheorem{theorem}{\n{Theorem}}[section]

\newtheorem{remark}{\n{Remark}}[section]
\newtheorem{lemma}{\n{Lemma}}[section]

\newtheorem{corollary}{\n{Corollary}}[section]

\newtheorem{proposition}{\n{Proposition}}[section]

\newcommand{\Rmnum}[1]{\expandafter\romannumeral #1}

\font\n=cmcsc12

  \textwidth = 420pt
\font\n=cmcsc10

\begin{document}
\title{\ Empirical Bayes, SURE and Sparse Normal Mean Models}
\author{\large Xianyang Zhang and Anirban Bhattacharya
\thanks{Department of Statistics, Texas A\&M University, College Station, TX 77843, USA. E-mail: zhangxiany@stat.tamu.edu; anirbanb@stat.tamu.edu}
\medskip\\
{\large Texas A\&M University}
}
\date{\normalsize This version: \today}
\maketitle
\sloppy%
\textbf{Abstract} This paper studies the sparse normal
mean models under the empirical Bayes framework. We focus on the mixture priors with an atom at zero and a density component centered at a data driven location determined by maximizing the marginal likelihood or
minimizing the Stein Unbiased Risk Estimate. We study the properties of the corresponding posterior median and posterior mean.
In particular, the posterior median is a thresholding rule and enjoys the multi-direction shrinkage property that shrinks the observation toward either the origin or the data-driven location.
The idea is extended by considering a finite mixture prior, which is flexible to model the cluster structure of the unknown means.
We further generalize the results to heteroscedastic normal mean models. Specifically, we propose a semiparametric estimator which can be
calculated efficiently by combining the familiar EM algorithm with
the Pool-Adjacent-Violators algorithm for isotonic regression. The effectiveness of our
methods is demonstrated via extensive numerical studies.
\\
\strut \textbf{Keywords:} EM algorithm, Empirical Bayes,
Heteroscedasticity, Isotonic regression, Mixture modeling, PAV
algorithm, Sparse normal mean, SURE, Wavelet

\section{Introduction}
A canonical problem in statistical learning is the compound estimation of (sparse) normal means from a single observation. The
observed vector $\mathbf{X}=(X_1,\dots,X_p)\in\mathbb{R}^p$ arises from the location model,
$$X_i=\mu_i+\epsilon_i,\quad \epsilon_i\sim^{i.i.d} N(0,1),$$
for $1\leq i\leq p,$ and the goal is estimating the unknown mean vector $(\mu_1,\dots,\mu_p)$ as well as recovering its support. This kind of problems arise in many different contexts such as
adaptive nonparametric regression using wavelets, multiple testing, variable selection and many other areas in statistics.
Location model also carries significant practical relevance in many statistical applications because the observed data are often understood,
represented or summarized as the sum of a signal vector and Gaussian errors.

In this paper, we tackle the problem from the empirical Bayes perspective which has seen a revival in recent years, see e.g. \citet[JS hereafter]{johnstone2004needles}, \citet{brown2009nonparametric,jiang2009general,koenker2014convex,martin2014asymptotically,petrone2014bayes}, among others. \citet{morris1983parametric} classified empirical Bayes into two types, namely parametric empirical Bayes and nonparametric empirical Bayes. 
In sparse models, the parametric (empirical) Bayes approach usually begins with a spike-and-slab prior on each $\mu_i$ that separates signals from noise, which includes the case when the spike component is a point mass at zero [see \citet{mitchell1988bayesian,george1993variable,ishwaran2005spike}]. In contrast, the nonparametric empirical Bayes approach assumes a fully nonparametric
prior on the means which is estimated by general maximum likelihood, resulting in an estimate which is a discrete distribution with no more than $p+1$ support points. Our strategy is different
from both the empirical Bayes with spike-and-slab priors and the general maximum likelihood empirical Bayes (GMLEB).
To account for sparsity, we impose a mixture prior on the entries of the mean vector which admits a point mass at zero. The {\em signal distribution}, that is, the distribution of the non-zero means, is modeled as a finite mixture distribution whose component densities could have nonzero centers.
Thus, the class of priors considered belong to an intermediate class between the spike-and-slab priors and the fully nonparametric priors. The finite mixture approach gives the flexibility of a nonparametric model while with the convenience of a parametric one, see e.g. \citet{allison2002mixture} and \citet{muralidharan2010empirical}.


\begin{figure}[h]
\centering
\includegraphics[height=6cm,width=6cm]{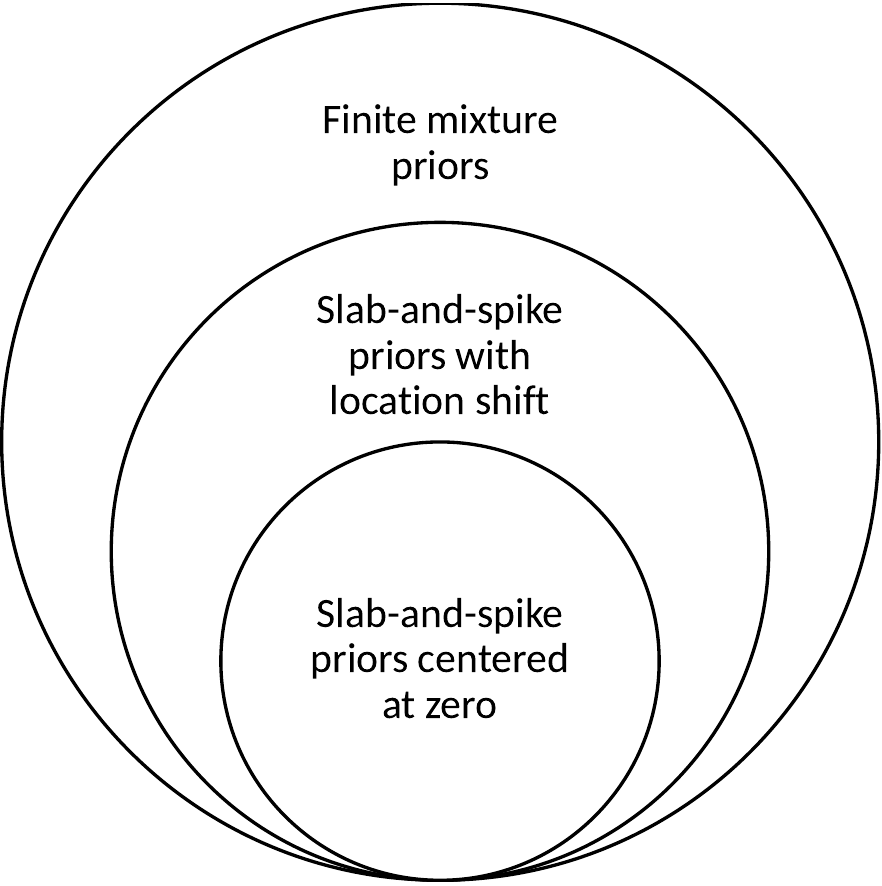}
\caption{Relationship among priors.}\label{fig:prior}
\end{figure}

One advantage of the proposed mixture prior is that it allows users to impose a point mass at zero, which implies sparsity in the posterior median or some other appropriate summary of the posterior \citep{raykar2011empirical}. However, such a goal is not easily achieved for the
GMLEB as its solution does not necessarily have a point mass at zero, and an additional thresholding step might be required to obtain a sparse solution.
Another salient feature of the proposed prior is its added flexibility in modeling potential cluster structures in the
nonzero entries. For example, the posterior mean and median associated with the proposed prior have a multi-direction shrinkage property that shrinks observation toward its nearest center (determined by data). By contrast, the posterior mean and median from usual spike-and-slab prior
shrinks datum toward zero regardless its distance from the origin (although the amount of shrinkage may decrease as the observation gets farther away from zero).
Focusing in particular on two-component mixture priors with a non-zero location parameter in the slab component, we provide an in-depth study of the properties of the posterior median, which is a thresholding rule and enjoys the two-directional shrinkage property. We show through numerical studies that inclusion of the location parameter (determined by the data) significantly improves the performance of the posterior median over JS (2004) when the nonzero entries exhibit certain cluster structure. It is also worth mentioning that the hyperparameters in the proposed prior can be estimated efficiently using the familiar EM-algorithm,
which saves considerable computational cost in comparison with the GMLEB. A price we pay here is the selection of the number of components in the mixture prior,
which can be overcome using classical model selection criterions such as the Bayesian information criterion \citep{fraley2002model}.

We also study the risk properties of the posterior mean under the mixture prior. We propose to estimate the hyperparameters
by minimizing the corresponding Stein Unbiased Risk Estimate (SURE). A uniform consistency result is proved to justify the theoretical validity of this procedure.
As far as we are aware, the use of SURE to tune the hyperparameters in the current context has not been previously considered in the literature.

We further extend our results to sparse heteroscedastic normal mean models, where the noise can have different variances. Heteroscedastic normal mean models have been recently studied from the empirical Bayes perspective; see \citet{xie2012sure,tan2015improved} and \cite{weinstein2015group}.
Our focus here is on the \emph{sparse} case which has not been covered by the aforementioned works. The proposed approach is different from existing ones in terms of the prior as well as
the way we tune the hyperparameters. Motivated by \citet{xie2012sure}, we propose a semiparametric approach to
account for the ordering information contained in the variances in estimating the means. To obtain the marginal maximum likelihood estimator (MMLE), we develop a modified EM algorithm that invokes the pool-adjacent-violators (PAV) algorithm in M-step, see more details in
Section \ref{sec:hete}.

The rest of the article is organized as follows. In Section \ref{sec:sparse1}, we begin with a formal introduction of the empirical Bayes procedure in the sparse normal mean models with two component mixture priors, where the density component has a (nonzero) location shift parameter.
Section \ref{sec:median} studies the posterior median. Extensions to finite mixture priors on the means are considered in Section \ref{sec:mix}. Section \ref{sec:mean} contains some results on the risk of the posterior mean and the uniform consistency for SURE.
Section \ref{sec:hete} concerns the heteroscedastic sparse normal mean models.
Section \ref{sec:sim} is devoted to numerical studies and empirical analysis of image data.
The technical details are gathered in the appendix.

\section{Sparse normal mean models}\label{sec:sparse}

\subsection{Two component mixture priors and the MMLE}\label{sec:sparse1}
Throughout the paper, we assume
that the mean vector $(\mu_1,\dots,\mu_p)$ is sparse in the
sense that many or most of its components are zero.
The notion of
sparseness can be captured by independent prior distributions on each $\mu_i$ given
by the mixture,
\begin{align}\label{eq-mix1}
f(\mu) = (1 - w) \delta_0(\mu) + w f_s(\mu), \quad w \in [0, 1],
\end{align}
where $f_s$ is a density on $\mathbb{R}$, and $\delta_0$ denotes a point mass at zero. While $f_s$ is allowed to be completely unspecified in GMLEB, we aim to harness additional structure by modeling $f_s$ in a semi-parametric way. To begin with, we model $f_s$ via a location-scale family $\gamma(., b, c)$ with scale parameter $b$ and location parameter $c$, i.e.,
$\gamma(\mu;b,c)=b\gamma_0(b(\mu-c))$ with
$\gamma_0(\mu)=\gamma(\mu;1,0)$ and $b>0$. Typical choices of
$\gamma$ include the double exponential or Laplace distribution,
\begin{align}\label{pior-exp}
\gamma(\mu;b,c)=\frac{1}{2}b\exp(-b|\mu-c|),
\end{align}
and the normal distribution
\begin{align}\label{pior-exp}
\gamma(\mu;b,c)=\frac{b}{\sqrt{2\pi}}\exp\{-b^2(\mu-c)^2/2\},
\end{align}
for $b>0$ and $c\in\mathbb{R}$. Note that the location parameter $c$
is equal to zero in the prior distribution suggested by JS (2004).
Our numerical results in Section \ref{sec:sim} suggest that location parameter, which captures cluster structure in signals, can play an important role in sparse normal
mean estimation.

Let
$g(x;b,c)=\int_{-\infty}^{+\infty}\phi(x-\mu)\gamma(\mu;b,c)d\mu$ be
the convolution of $\phi(\cdot)$ and $\gamma(\cdot;b,c)$, where
$\phi(\cdot)$ denotes the standard normal density. Under (\ref{eq-mix1}), the marginal
distribution for $X_i$ is
$$m(x;w,b,c)=(1-w)\phi(x)+wg(x;b,c),$$
and the corresponding posterior distribution for $\mu_i$ is equal to
$$\pi(\mu|X_i=x,w,b,c)=(1-\alpha(x))\delta_0(\mu)+\alpha(x)h(\mu|x,b,c),$$
where
$$\alpha(x)=\frac{wg(x;b,c)}{m(x;w,b,c)}\quad \text{and}\quad h(\mu|x,b,c)=\frac{\phi(x-\mu)\gamma(\mu;b,c)}{g(x;b,c)}.$$

In the sequel, we proceed to estimate the parameters $(w,b,c)$ by maximizing the marginal likelihood of $\mathbf{X}$.
Specifically, the MMLE $(\hat{w},\hat{b},\hat{c})$ is defined as
\begin{align}\label{margin}
(\hat{w},\hat{b},\hat{c})=\arg\max
\sum^{p}_{i=1}\log\{(1-w)\phi(X_i)+wg(X_i;b,c)\},
\end{align}
where the optimization is subject to the constraints that $b>0$,
$-\max_{1\leq i\leq p}|X_i|\leq c\leq \max_{1\leq i\leq p}|X_i|$, and $0\leq
w\leq 1$.
The optimization problem (\ref{margin}) can be solved efficiently
using the EM algorithm. 

\subsection{The posterior median}\label{sec:median}
In case of $c = 0$, JS (2004) noted that the median of the posterior
distribution $\pi(\mu_i|X_i=x,w,b,c)$, denoted by $\delta(x;w,b,c)$,
has the thresholding property, that is, the posterior median is exactly zero on a symmetric interval around the origin. The thresholding property continues to hold even when $c \ne 0$, whence there exist positive constants $t_1(w, b, c)$ and $t_2(w, b, c)$ such that $\delta(x;w,b,c) = 0$ for any $-t_2(w, b, c) \le x \le t_1(w, b, c)$. For $c\neq 0$, the thresholding
levels $t_1(w,b,c)$ and $t_2(w,b,c)$ are unequal, which results in
an asymmetric thresholding rule, see Proposition \ref{prop1}. This
is in sharp contrast with the case $c=0$, where the posterior median
is antisymmetric, i.e., $\delta(-x;w,b,0)=-\delta(x;w,b,0)$ [see
Lemma 2 of JS (2004)]. Figures
\ref{fig:delta} plots the posterior median
$\delta(x;w,b,c)$ as a function of $x$ for various values of $c$.
For $c\neq 0$, the posterior median enjoys the so-called
two-direction shrinkage property i.e., when $x$ is close to
zero, it is being shrunk toward the origin; when $x$ is close to
$c$, it is being pulled toward $c$.

\begin{figure}[h]
\centering
\includegraphics[height=6cm,width=4cm]{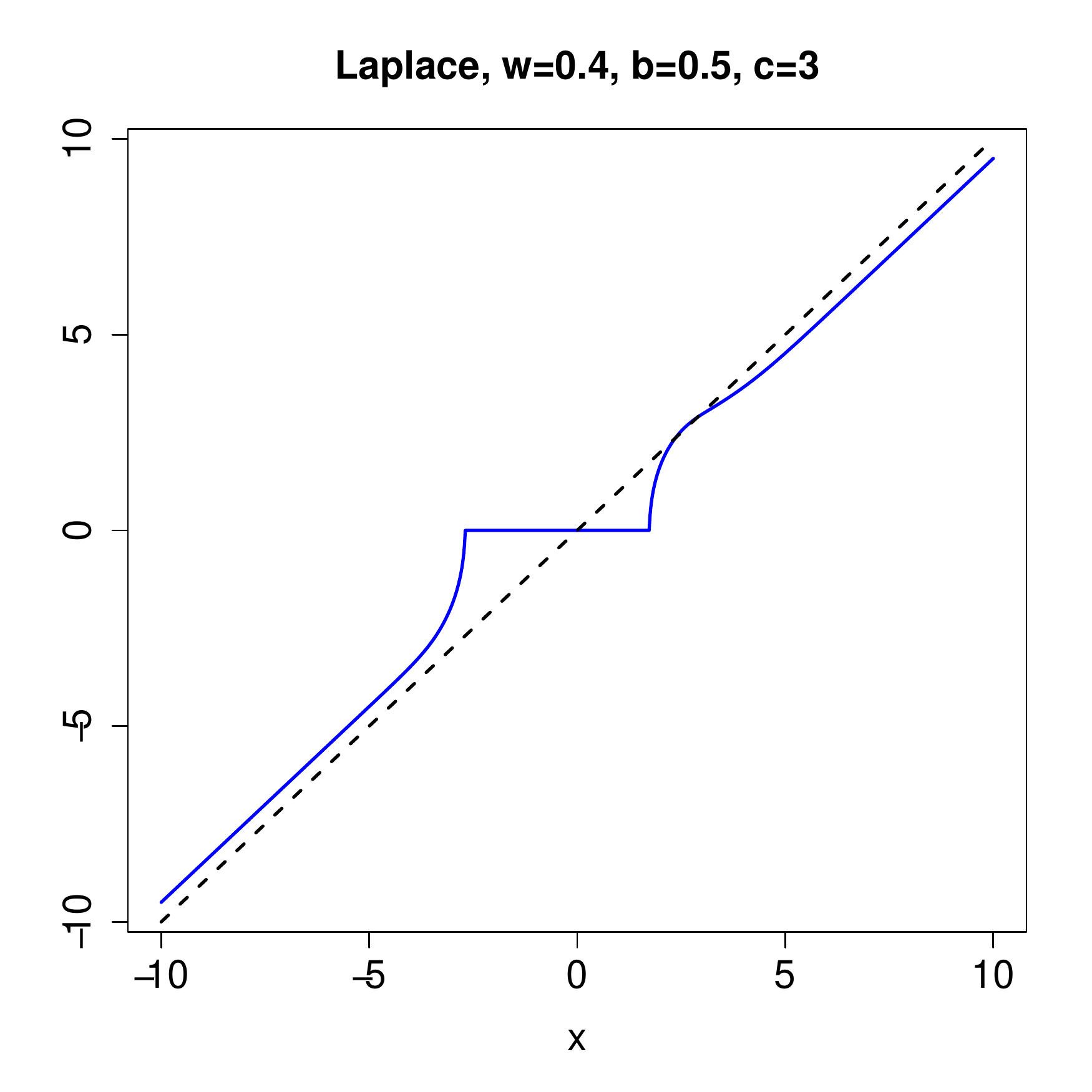}
\includegraphics[height=6cm,width=4cm]{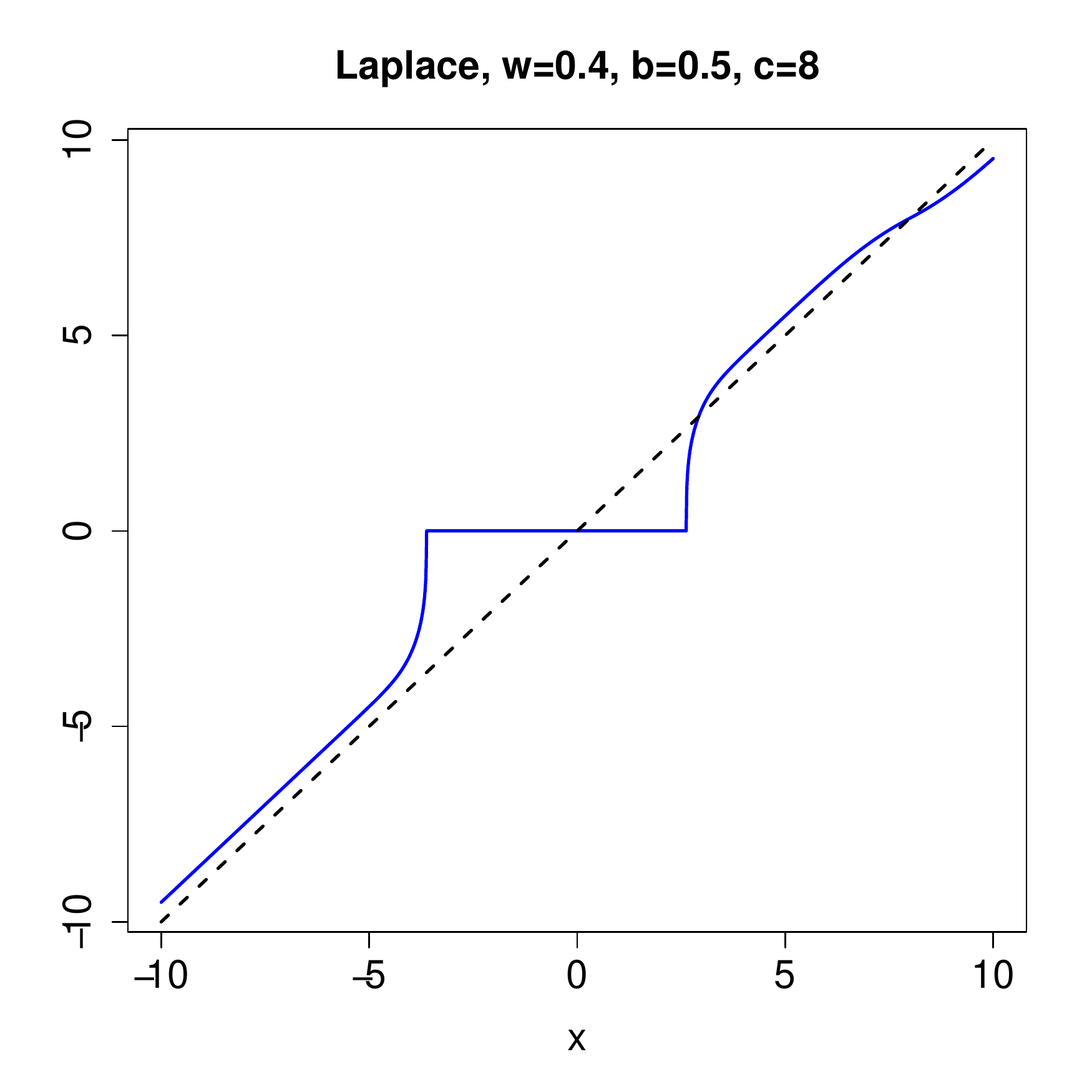}
\includegraphics[height=6cm,width=4cm]{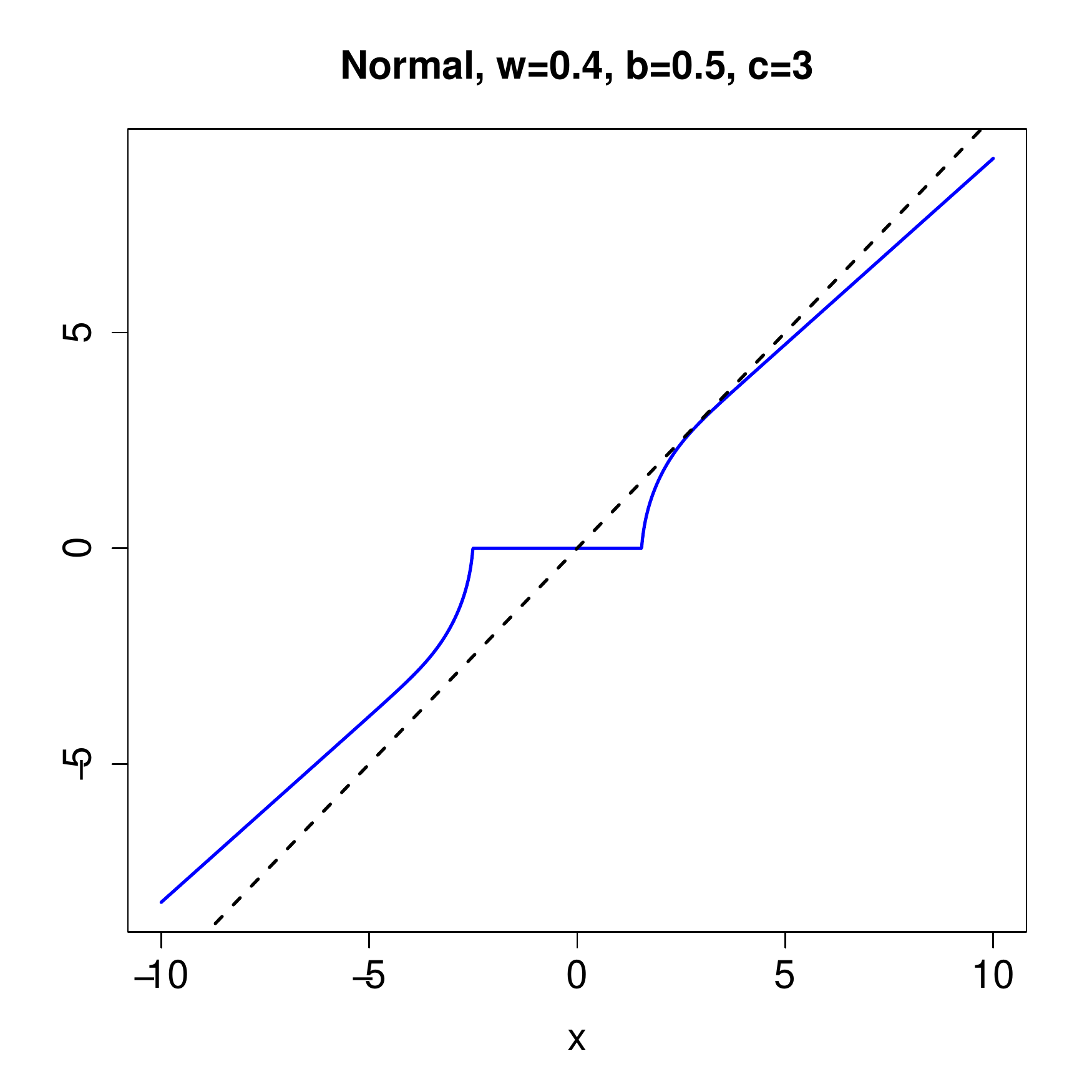}
\includegraphics[height=6cm,width=4cm]{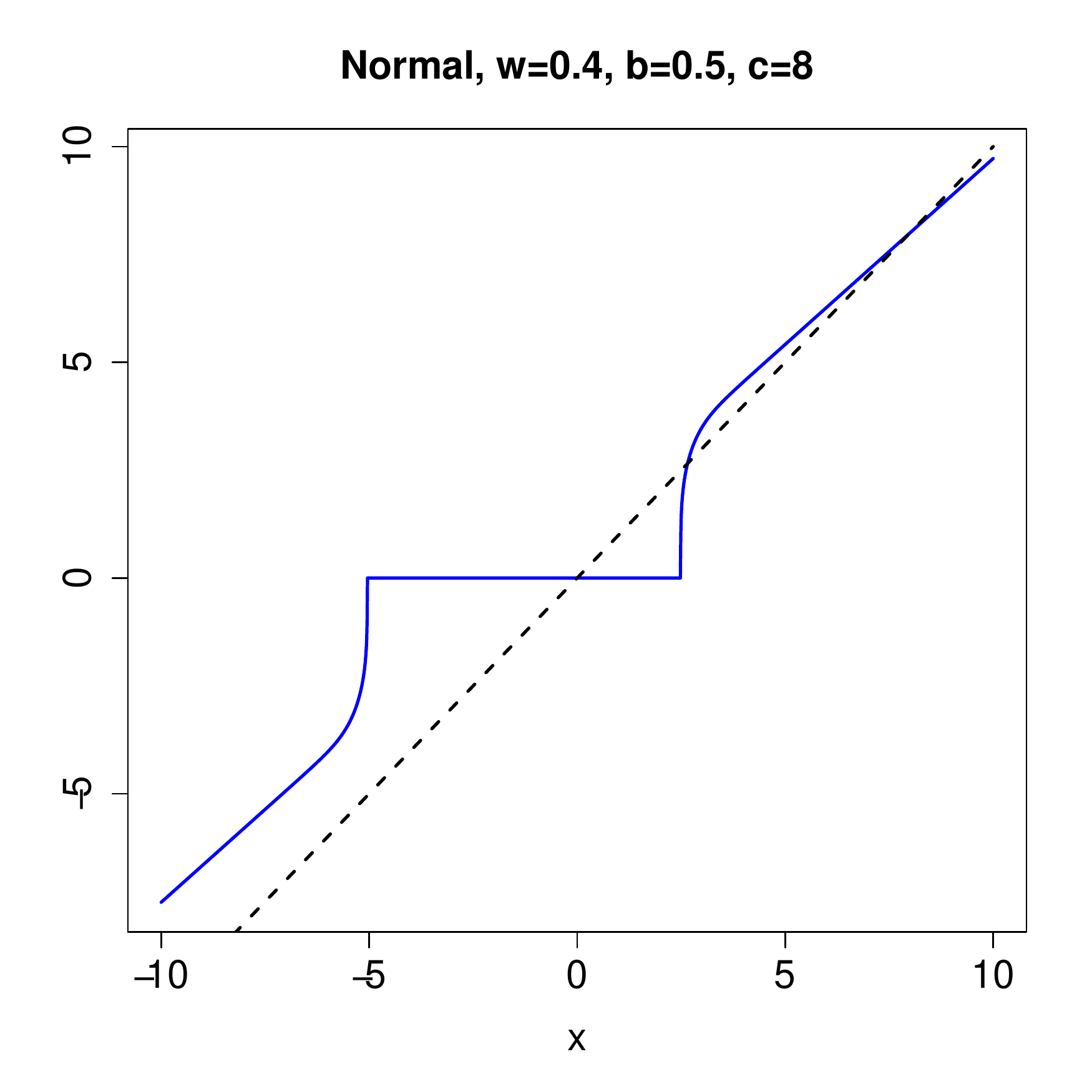}
\caption{Posterior median function for $w=0.4$, $b=0.5$, and
$c=3,8$, where the prior density component is double exponential or
normal with the location parameter $c$.}\label{fig:delta}
\end{figure}

We present some properties regarding the
posterior median below. For the sake of clarity, we set $b=1$ and write
$\gamma(\mu;c)=\gamma(\mu;1,c)$, $\delta(x;w,c)=\delta(x;w,1,c)$,
and $g(x;w,c)=g(x;w,1,c)$. The results can be extended to the general case by rescaling $x$ and $\mu$.
\begin{proposition}\label{prop1}
Assume that there exit $\Lambda,M>0$ such that
\begin{align}\label{eq-lem2}
\sup_{u>M}\left|\frac{d}{du}\log\gamma(u)\right|\leq \Lambda.
\end{align}
The posterior median $\delta(x;w,c)$ satisfies the following properties.
\\(1) $\delta(x;w,c)$ is a nondecreasing function of $x$;
\\(2) Suppose
\begin{equation}\label{eq-prop1}
\begin{split}
\left|\int_{0}^{+\infty}\phi(\mu)\gamma_0(\mu-c)d\mu-\int_{-\infty}^{0}\phi(\mu)\gamma_0(\mu-c)d\mu\right|\leq \frac{1-w}{\sqrt{2\pi}w}.
\end{split}
\end{equation}
Then there exist $t_1:=t_1(w,c)\geq 0$ and $t_2:=t_2(w,c)\geq 0$ such that
\begin{align}
&
\int_{0}^{+\infty}\phi(t_1-\mu)\gamma(\mu;c)d\mu=(1-w)\phi(t_1)/(2w)+g(t_1;c)/2,\label{median1}
\\&
\int_{-\infty}^{0}\phi(-t_2-\mu)\gamma(\mu;c)d\mu=(1-w)\phi(t_2)/(2w)+g(-t_2;c)/2,\label{median2}
\end{align}
and
\begin{align*}
\delta(x;w,c)\begin{cases}
  <0, & \mbox{if }~~x<-t_2,\\
  =0, & \mbox{if }~~-t_2\leq x\leq t_1,\\
  >0, & \mbox{if }~~x>t_1.
\end{cases}
\end{align*}
(3) $|\delta(x;w,c)|\leq |x|\vee |c|$ for any $0\leq w\leq 1$ and $c$;
\\(4) Under (\ref{eq-prop1}), $|\delta(x;w,c)-x|\leq t_1(w,c)\vee t_2(w,c)+c+c_0$ for some constant $c_0>0$.
\end{proposition}
\begin{remark}
{\rm In the case of double exponential prior with $b=1$ and $c>0$, the threshold levels $t_1$ and $t_2$, and the weights and location parameter are related by
\begin{align*}
&\frac{1}{w}+\beta(t_1;c)=e^c\frac{\Phi(t_1-1-c)}{\phi(t_1-1)}+e^{-c}\frac{\Phi(c-t_1-1)-\Phi(-t_1-1)}{\phi(t_1+1)},\\
&\frac{1}{w}+\beta(-t_2;c)=e^{-c}\frac{\Phi(c+t_2-1)}{\phi(t_2-1)},
\end{align*}
where $\beta(t;c)=g(t;c)/\phi(t)-1.$ See Figure \ref{fig:t}.
}
\end{remark}

\begin{figure}[h]
\centering
\includegraphics[height=6cm,width=6cm]{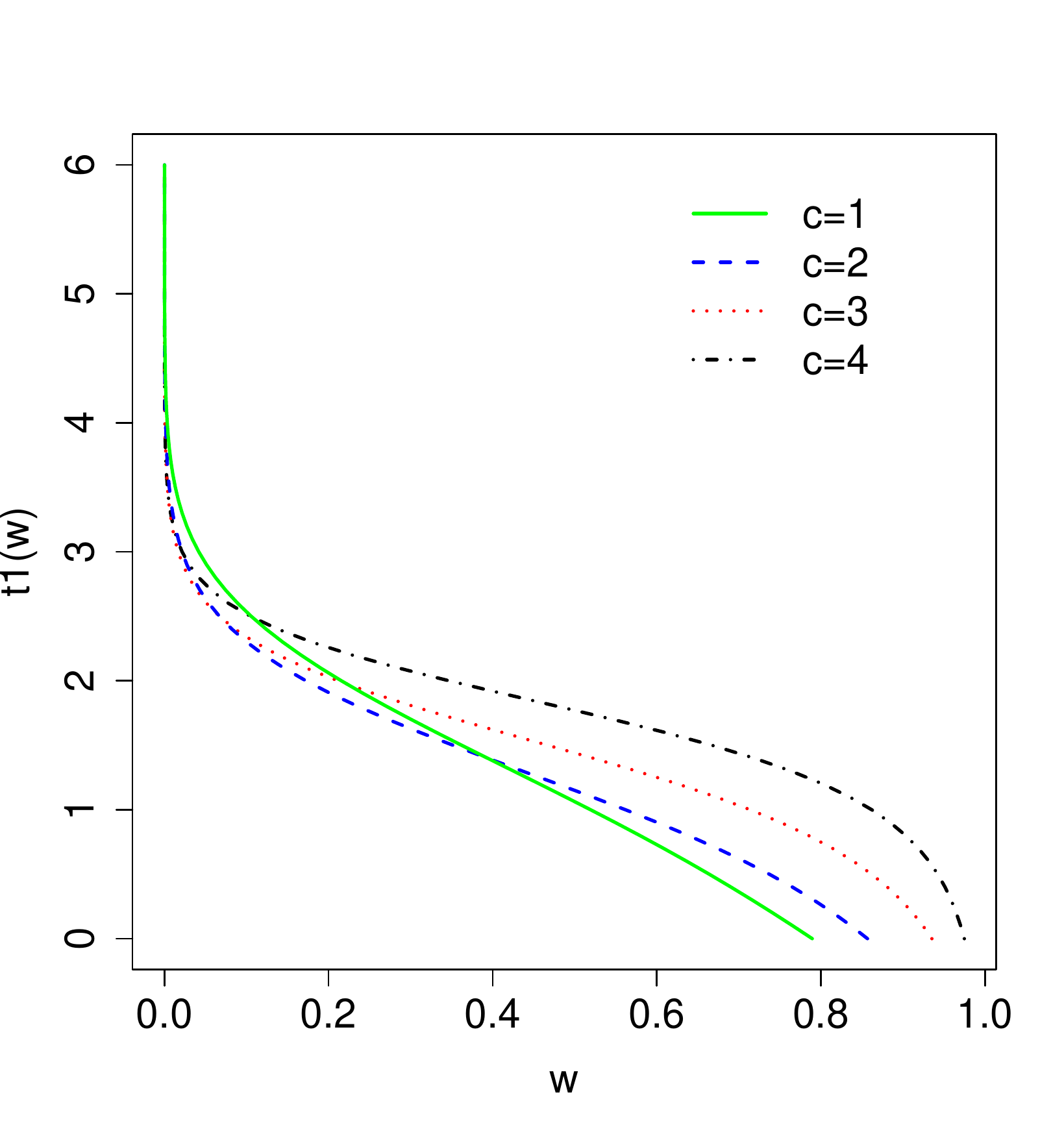}
\includegraphics[height=6cm,width=6cm]{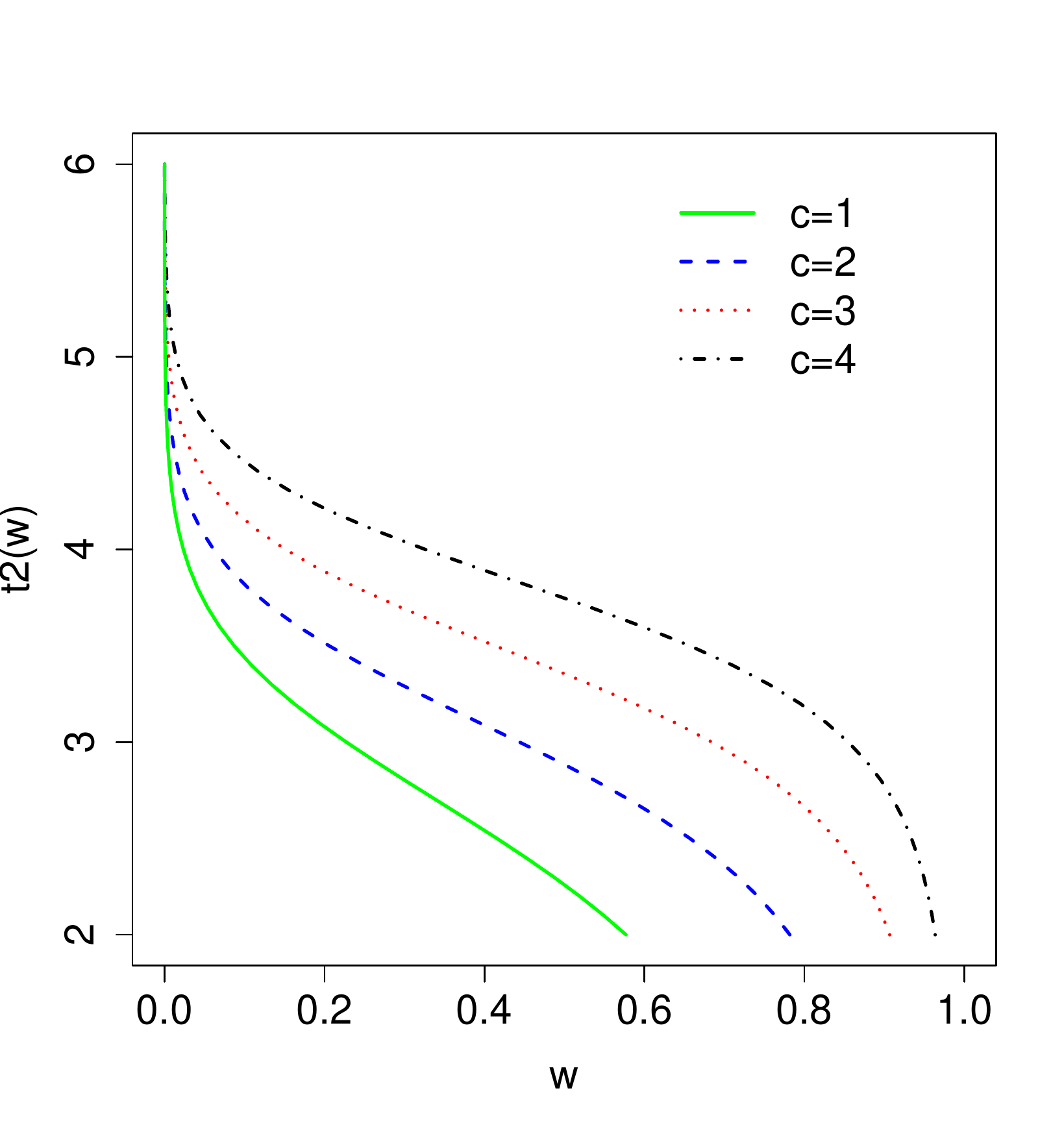}
\caption{The threshold levels $t_1(w;c)$ and $t_2(w;c)$ as functions of non-zero prior mass $w$ for the double exponential density with $b=1$ and $c=1,2,3,4$.}\label{fig:t}
\end{figure}

The results in Proposition \ref{prop1} are applicable to the double
exponential prior with location shift. However, Condition
(\ref{eq-lem2}) requires the tails of $\gamma$ to be exponential or heavier and thus rules out the Gaussian prior. In Section
\ref{closed-form}, we provide the closed-form representations for
$\delta(x;w,b,c)$ when $\gamma$ is double exponential or normal.
Based on the explicit expressions, we obtain the following results
for double exponential and Gaussian priors which reflect their
different tail behaviors.
\begin{lemma}\label{lemma:double}
When $\gamma$ is double exponential, the posterior median
$\delta(x;w,b,c)$ has the following properties:
\\(1) $\delta(c;w,b,c)-c\rightarrow 0$ as $c\rightarrow +\infty$.
\\(2) $\delta(x;w,b,c)-(x-b)\rightarrow 0$ as $x-c\rightarrow
+\infty$ and $x\rightarrow +\infty.$
\\(3) $\delta(x;w,b,c)-(x+b)\rightarrow 0$ as $x-c\rightarrow
-\infty$ and $x\rightarrow -\infty.$
\end{lemma}
Property (1) shows that there is no shrinkage effect for the posterior median when $x=c$;
Properties (2)-(3) suggest that the posterior median
becomes a shrinkage rule as $|x|\rightarrow +\infty$. In other words, the effect of the atom at zero and the impact of $c$ are both negligible as $|x|\rightarrow+\infty.$
\begin{lemma}\label{lemma:normal}
When $\gamma$ is normal, we have
$$\delta(x;w,b,c)-\frac{
x/b^2+c}{1/b^2+1}\rightarrow 0,\quad \text{as}\quad |x|\rightarrow +\infty.$$
\end{lemma}
We note that $(x/b^2+c)(1/b^2+1)$ is the posterior mean when $w=1$.
Intuitively, when $c$ is close to the center of the nonzero
components, $\delta(x;w,b,c)$ enjoys the property by shrinking $x$
toward $c$, which may lead to further risk reduction as compared to
the thresholding rules considered in JS (2004).

We would like to point out that the posterior median
resulting from the prior with location-shift density component
defines a new class of thresholding rules i.e., $\delta(x;w,b,c)$.
By (2) of Proposition \ref{prop1}, there exist two positive numbers $t_1$ and $t_2$ such
that $\delta(x;w,b,c)=0$ if and only if $-t_2\leq x\leq t_1$. Also
$\delta(x;w,b,c)$ is strictly increasing for $x>t_1$ and $x<-t_2$.
Thus, the inverse function $\delta^{-1}(t;w,b,c)$ is defined for
any $t\neq 0$. Define the penalty function,
\begin{align}
\mathcal{P}(\theta;w,b,c)=\begin{cases}
\int^{\theta}_{0}(\delta^{-1}(t;w,b,c)-t)dt \quad
& \text{if} \quad  \theta \neq 0, \\
0 \quad & \text{if} \quad  \theta=0.
\end{cases}
\end{align}

\begin{figure}[h]
\centering
\includegraphics[height=6cm,width=6cm]{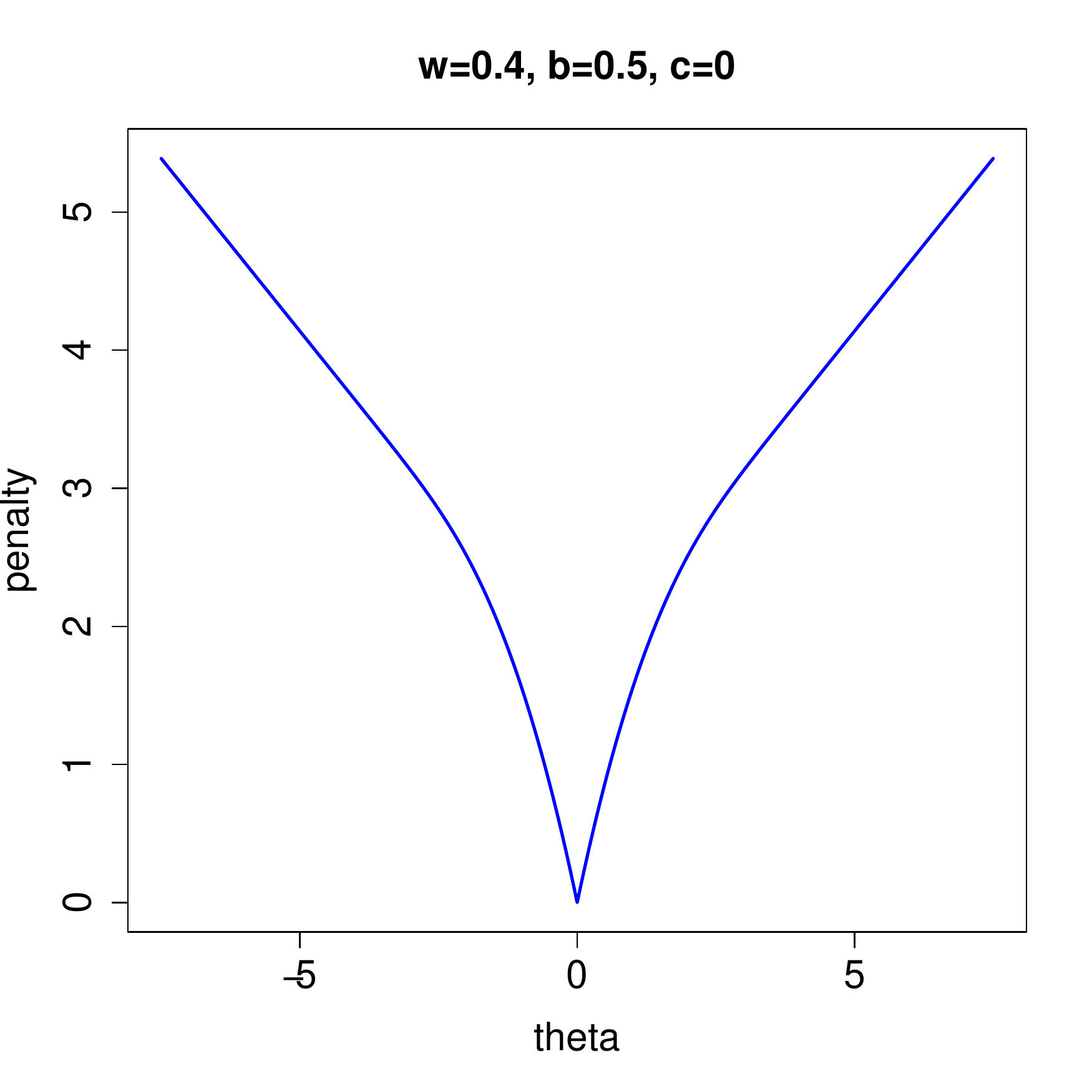}
\includegraphics[height=6cm,width=6cm]{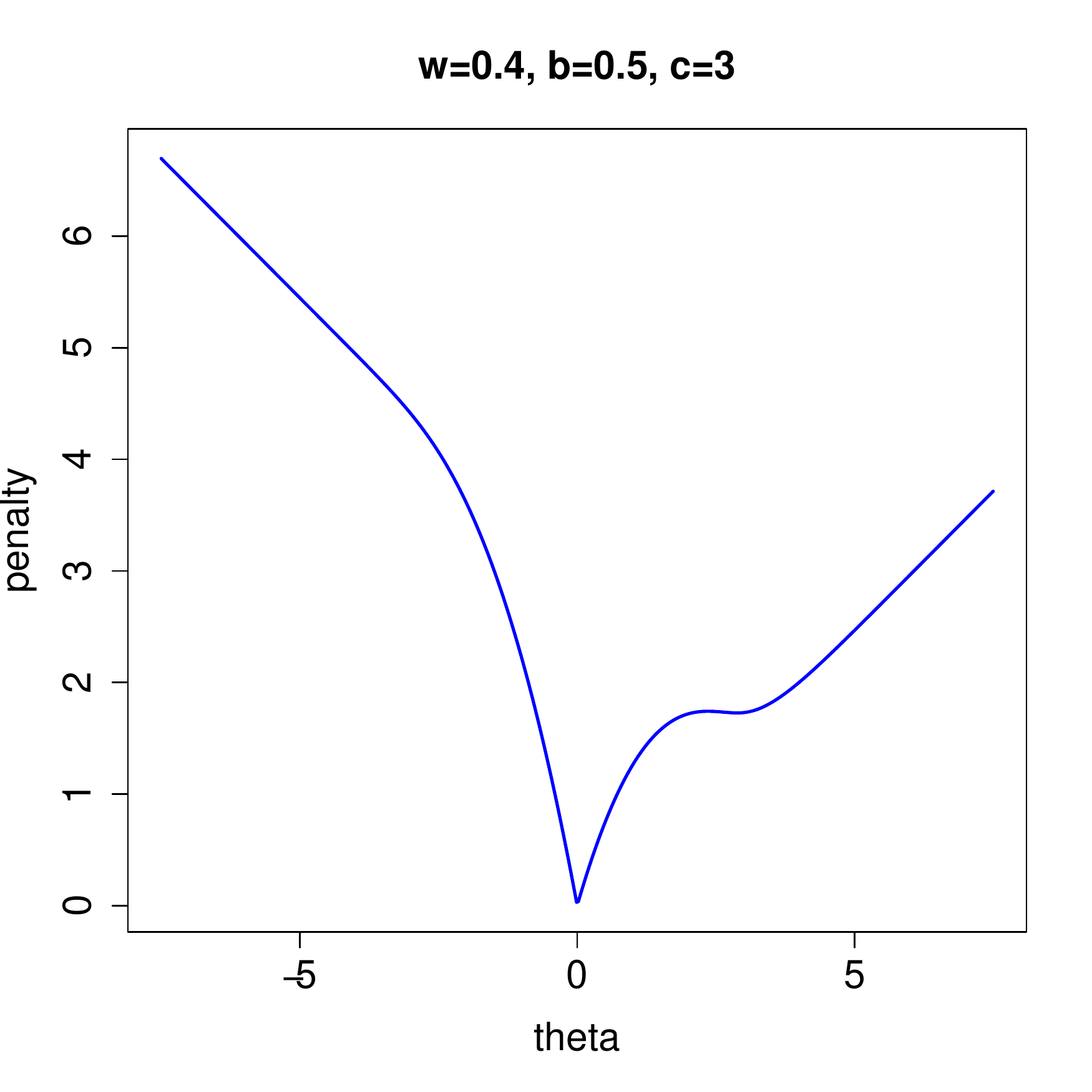}
\caption{Penalty function for $w=0.4$, $b=0.5$ and
$c=0,3$, where the prior density component is double
exponential.}\label{fig:pen}
\end{figure}

Consider the optimization problem
\begin{align}\label{pen}
\hat{\theta}=\hat{\theta}(x;w,b,c):=\arg\min_{\theta}\frac{1}{2}(x-\theta)^2+\mathcal{P}(\theta;w,b,c).
\end{align}
In the appendix, we prove that the solution to (\ref{pen}) is $\delta(x;w,b,c)$.
\begin{lemma}\label{lemma:con}
$\hat{\theta}(x;w,b,c)=\delta(x;w,b,c)$.
\end{lemma}
Figure \ref{fig:pen} plots the penalty function
$\mathcal{P}(\theta;w,b,c)$, where $\delta(x;w,b,c)$ is the
posterior median associated with the double exponential prior with
$w=0.4$, $b=0.5$ and $c=0,3$. Compared to commonly used penalties, the penalty function here is nonstandard in the sense that it is asymmetric about zero, and is non-monotonic over $[0,+\infty)$.
It is of interest to study the penalized regression problem based on the new penalty function
$\mathcal{P}(\theta;w,b,c)$, and employ the empirical Bayes method
to select the tuning parameters $(w,b,c)$. We leave this topic to future research.

\begin{remark}
{\rm We remark that the relationship between penalty function and its solution in location model as described in (\ref{pen}) holds for commonly used penalty functions such as
Lasso, SCAD \citep{fan2001variable} and MCP \citep{zhang2010nearly}. }
\end{remark}

\begin{remark}
{\rm Besides the posterior median, a general class of Bayes
thresholding rule which combines the soft and hard thresholding
rules can be obtained by minimizing a mixture loss combining the
$l_p$ loss (for $p>0$) and the $l_0$ loss for the posterior
distribution. See more details in \cite{raykar2011empirical}. }
\end{remark}

\subsection{Finite mixture priors}\label{sec:mix}
A natural extension to pursue here is to replace the density component $\gamma$ by a finite mixture distribution, which can be used to model the cluster structure of the nonzero means
[see \citet{muralidharan2010empirical}]. Specifically, one can model $f_s$ in \eqref{eq-mix1} as a finite mixture distribution and
consider
the prior of the form
$$f(\mu,\theta)=w_0\delta_0(\mu)+\sum^{d}_{j=1}w_j\gamma(\mu;b_j,c_j),$$
with $w_j\geq 0$ and $\sum_{j=0}^{d}w_j=1$, and
$\theta=(w_0,w_1,b_1,c_1,\dots,w_{d},b_{d},c_{d})$. Let
$$L(d,\theta)=\sum_{i=1}^{p}\log\left\{w_0\phi(X_i)+\sum^{d}_{j=1}w_jg(X_i;b_j,c_j)\right\}$$
be the log-marginal likelihood. In this case,
the MMLE is defined as
\begin{align}\label{eq-mix2}
\hat{\theta}:=(\hat{w}_0,\hat{w}_1,\hat{b}_1,\hat{c}_1,\dots,\hat{w}_{d},\hat{b}_{d},\hat{c}_{d})=\arg\max_{\theta} L(d;\theta),
\end{align}
subject to the constraints that $0\leq w_j\leq 1$,
$\sum^{d}_{j=0}w_j=1$, $b_j\geq 0$, and $-\max_{1\leq i\leq
p}|X_i|\leq c_j\leq \max_{1\leq i\leq p}|X_i|$ for $1\leq j\leq d$.
As before, the solution to (\ref{eq-mix2}) can be
obtained using the familiar EM algorithm.
A sparse estimator for $\mu_i$ is given by the posterior median
$\delta(X_i,\hat{\theta})$, which is again a thresholding rule and
has the multi-direction shrinkage property in the sense that
it pulls $X_j$ toward one of the data driven locations $\hat{c}_j$
when $X_j$ is away from zero, see Figure \ref{fig:delta-normal2}.
\begin{figure}[h]
\centering
\includegraphics[height=6cm,width=6cm]{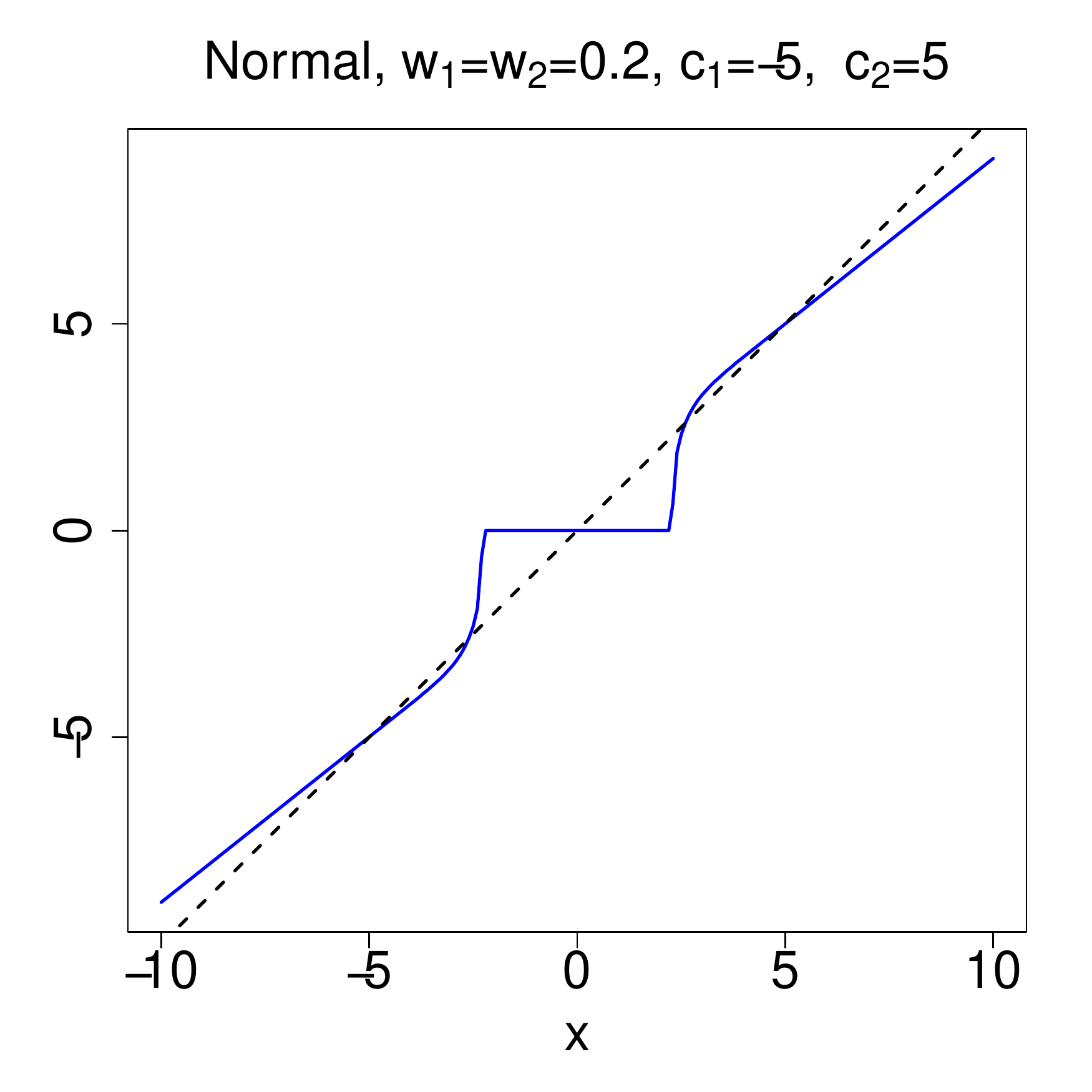}
\includegraphics[height=6cm,width=6cm]{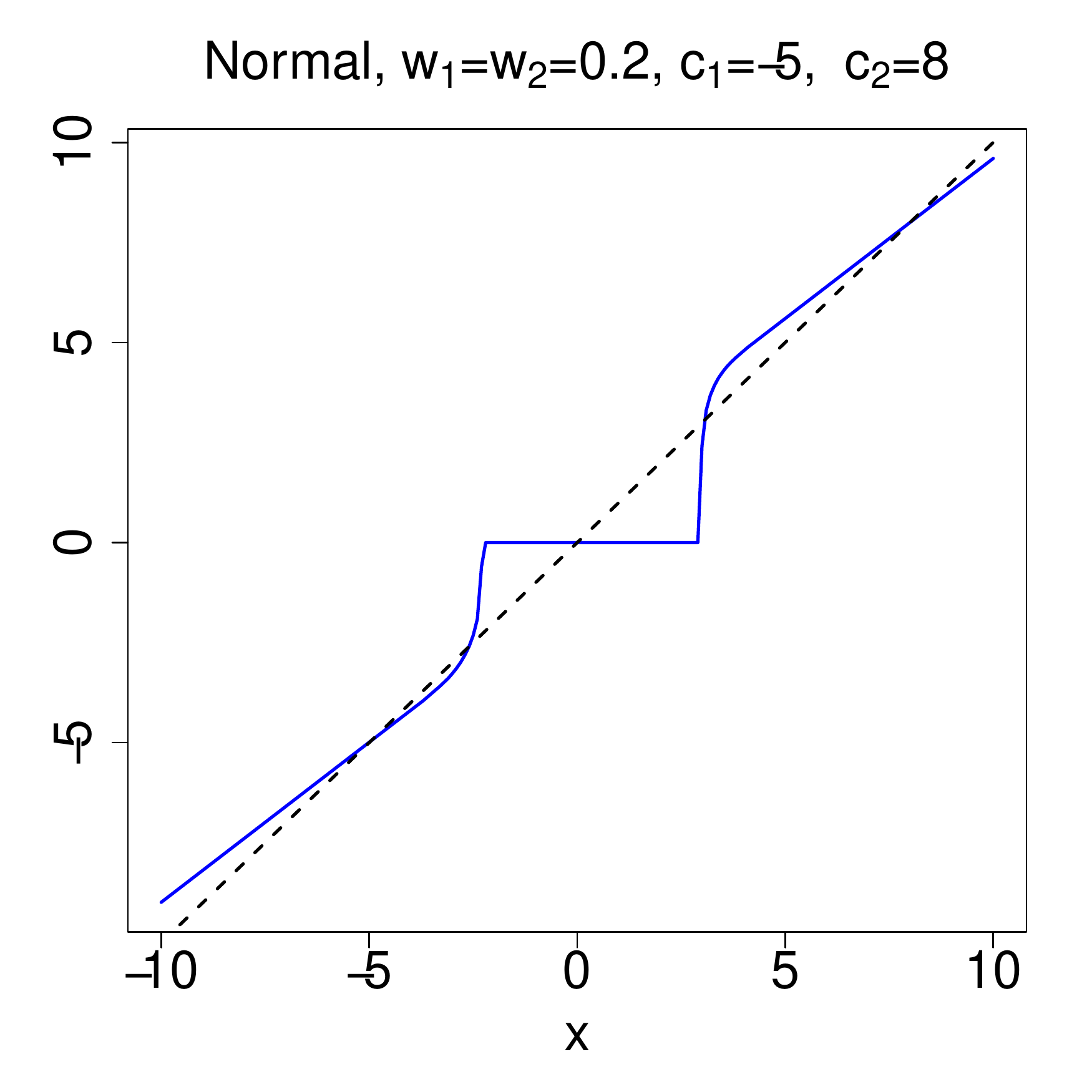}
\caption{Posterior median function for $d=2$, where the prior density component is normal mixture with the
location parameters $c_1$ and $c_2$.}\label{fig:delta-normal2}
\end{figure}

In practice, the number of mixture components is often unknown. In the sparse regime, $d$ is typically chosen as a relatively small number
to model the cluster structure of the nonzero entries. For example,
with $d=2$ and the constraint that $c_1<0<c_2$, the two density
components are designed to model the negative and positive signals
separately. Alternatively one can choose the number of clusters using the Bayesian information criterion [see e.g.
\citet{fraley2002model}]. Specially, the choice of $\hat{d}$ for $d$
maximizes
\begin{equation}\label{bic}
L(d;\hat{\theta})-3\log(p)d/2,
\end{equation}
over $1\leq d\leq M_0$, where $M_0$ is a pre-chosen upper bound. \citet{leroux1992consistent} proved that model selection based on a comparison of
BIC values does not underestimate the number of components; \citet{keribin1998consistent} and \citet{gassiat2013consistent} showed that BIC is consistent for selecting the number of
components.


\subsection{The posterior mean and SURE}\label{sec:mean}
We have so far focused on the posterior median which is a
thresholding rule. In this subsection, we turn to the posterior
mean which is no longer a thresholding rule but enjoys the same
multi-direction shrinkage property as the posterior median does. We
shall follow the setup in Section \ref{sec:mix}. Write
$g_j(x)=g(x;b_j,c_j)$ for $0\leq j\leq d$ with $g_0(x)=\phi(x)$. Let
$m(x)=\sum^{d}_{j=0}w_jg_j(x)$ and
$\rho_j(x)=w_jg_j(x)/m(x)$ for $0\leq j\leq d$. By Tweedie's formula, the posterior mean can be written as
\begin{align*}
\zeta(x,\theta)=x+\nabla \log m(x),
\end{align*}
where $\nabla=\partial/\partial x$.


Below we briefly discuss Stein's unbiased risk estimator (SURE; \citet{stein1981estimation}) for the posterior mean.
A function is said to be almost
differentiable if it can be represented by well-defined integral of
its almost-everywhere derivative. The following result was obtained
by \citet{george1986minimax} based on Stein's lemma.
\begin{theorem}\label{thm1}
Suppose $g_j$ and $\nabla g_j$ are both almost differentiable. If
\begin{align}
& \E |\nabla^2 g_j(X_i)/g_j(X_i)|<\infty,\quad \E (\nabla \log
g_j(X_i))^2<\infty, \label{eq2}
\end{align}
for $0\leq j\leq d$ and $1\leq i\leq p$. Then the squared error risk can be
expressed as
\begin{align*}
R(\theta) :=  \E\sum^{p}_{i=1}(\mu_i-\zeta(X_i,\theta))^2=p-\E
\sum^{p}_{i=1}D(X_i),
\end{align*}
where $D(X_i)=\sum^{d}_{j=0}\rho_j(X_i)D_j(X_i)-\sum_{0\leq j<k\leq
d}\rho_j(X_i)\rho_k(X_i)(\zeta_j(X_i)-\zeta_k(X_i))^2$
with $D_j(X_i)=(\nabla \log g_j(X_i))^2-2\nabla^2 g_j(X_i)/g_j(X_i)$
and $\zeta_j(X_i)=X_i+\nabla \log m_j(X_i)$.
\end{theorem}
Clearly, $\hat{R}(\theta) = p-
\sum^{p}_{i=1}D(X_i)$ is an unbiased estimator of the risk $R(\theta)$, which we shall refer to as SURE henceforth.
Note that $\sum^{p}_{i=1}D(X_i)$ is an unbiased estimator of the amount
of risk reduction offered by the posterior mean over the MLE
$\mathbf{X}$. When the prior is a normal mixture, the posterior mean
has the form of
\begin{align*}
\zeta(X_i,\theta)=\sum^{d}_{j=0}\rho_j(X_i)\zeta_j(X_i,\theta),\quad
\zeta_j(X_i,\theta)=\frac{X_i/b_j^2+c_j}{1/b_j^2+1},
\end{align*}
and
\begin{align*}
D_j(X_i,\theta)=\frac{2}{1/b_j^2+1}-\frac{(X_i-c_j)^2}{(1/b_j^2+1)^2},
\end{align*}
where $c_0=0$ and $b_0=\infty.$ Recall that $\rho_j(X_i)$ is the posterior probability
that $X_i$ is from the $j$th component of the mixture model. When $\rho_j(X_j)\gg \rho_k(X_i)$ for $k\neq j$, $\zeta_j(X_i,\theta)$ dominates in $\zeta(X_i,\theta)$ and thus $X_i$ is shrunk toward $c_j$.

As a consequence of Theorem \ref{thm1}, we obtain an explicit expression for $D(X_i,\theta)$.
\begin{corollary}
When the prior follows a normal mixture distribution, the unbiased estimator for the risk reduction is given by
\begin{align*}
D(X_i,\theta)=&\sum^{d}_{j=0}\rho_j(X_i)\left\{\frac{2}{1/b_j^2+1}-\frac{(X_i-c_j)^2}{(1/b_j^2+1)^2}\right\}
\\&-\sum_{0\leq j<k\leq
d}\rho_j(X_i)\rho_k(X_i)\left(\frac{X_i/b_j^2+c_j}{1/b_j^2+1}-\frac{X_i/b_k^2+c_k}{1/b_k^2+1}\right)^2.
\end{align*}
\end{corollary}
The first term in $D(X_i,\theta)$ measures the goodness of fit of the mixture model to the data, while the second term penalizes the pairwise distance between any two posterior means (with respect to the prior $\gamma_j$) weighted by the corresponding posterior probabilities $\rho_j$ and $\rho_k$.
In fact, maximizing the objective function $\sum_{i=1}^{p}D(X_i,\theta)$ results in an estimate for the hyperparameters $\theta$, i.e.,
\begin{align}\label{eq-stein}
\hat{\theta}=\arg\max_{\theta}\sum_{i=1}^{p}D(X_i,\theta).
\end{align}
In our simulations, we use the constrained version of the quasi-Newton BFGS (Broyden, Fletcher, Goldfarb and Shanno) method with multiple initial points to solve (\ref{eq-stein}).

To study the properties of $\hat{\theta}$, we shall focus on the case of two component mixture, i.e., $d=1$. Let $w=1-w_0$, $b=b_1$ and $c=c_1.$ To simplify the arguments, we set $b=1$, and
write $g(x;c)=g(x,1,c)$ and $m(x;\theta)=m(x;w,1,c)=m(x;w,c)$. 
Given $\theta=(w,c)$, let $\zeta(X_i;\theta)=X_i+\nabla \log
m(X_i;\theta)$. We state our main result below. For $a_1,a_2,a_3>0,$ denote by $\Theta:=\Theta(a_1,a_2,a_3)=\{(w,c): w\in [1/(a_1p^{a_2}),1]$ and $|c|\leq
a_3\log(p)\}$.

\begin{theorem}\label{thm-main}
Suppose $\gamma_0$ is unimodal and
\begin{equation}\label{thm-c1}
\sup_{u}\left|\nabla^j \log\gamma_0(u)\right|\leq \Lambda \quad
{a.e.},
\end{equation}
for $j=1,2$. Moreover, assume that
\begin{align}\label{eq-ass}
|\nabla^2 \log\gamma_0(u)-\nabla^2 \log\gamma_0(u')|\leq C
|u-u'|\quad {a.e.},
\end{align}
for some constant $C>0.$ Then we have uniformly for $(\mu_1,\dots,\mu_p)\in\mathbb{R}^p$,
\begin{align*}
\max_{(w,c)\in\Theta}p^{-1}|\hat{R}(w,c)-\E\hat{R}(w,c)|=O_p\left(\frac{(\log(p))^{3/2}}{\sqrt{p}}\right).
\end{align*}
The same conclusion holds when $\gamma_0$ is
double exponential.
\end{theorem}

Let $(\hat{w},\hat{c})=\arg\min_{(w,c)\in \Theta}\hat{R}(w,c)$ and $(\tilde{w},\tilde{c})=\arg\min_{(w,c)\in\Theta}R(w,c)$ with $R(w,c)=\E[\hat{R}(w,c)]$. As a consequence of Theorem \ref{thm-main}, we have
\begin{align*}
\hat{R}(\hat{w},\hat{c}) -  R(\tilde{w},\tilde{c})=& \hat{R}(\hat{w},\hat{c}) - \hat{R}(\tilde{w},\tilde{c}) + \hat{R}(\tilde{w},\tilde{c}) -  R(\tilde{w},\tilde{c})\leq  \hat{R}(\tilde{w},\tilde{c}) -  R(\tilde{w},\tilde{c})
\\ \leq&  \sup_{(w,c)\in\Theta}|\hat{R}(w,c) -  R(w,c)|=O_p\left(\frac{(\log(p))^{3/2}}{\sqrt{p}}\right).
\end{align*}

\begin{remark}
{\rm
A similar result as in Theorem \ref{thm-main} can be obtained for the Gaussian prior, whose proof involves the use of Gaussian concentration inequality for lipschitz functions. An additional assumption on the $\ell_2$ norm
of the mean vector is needed in this case. In our simulations, SURE based on the Gaussian prior performs as well as
the one based on the double exponential prior.
}
\end{remark}

\begin{remark}
{\rm
Consider the $\ell_q$ ball
$$B_{q}(\eta)=\left\{(\mu_1,\dots,\mu_p): p^{-1}\sum_{i=1}^{p}|\mu_i|^q\leq \eta^q\right\},$$
with small radius $\eta.$ The minimax risk under the squared loss is given by
$r_{q,2}=\eta^{2}$ for $q=2$ and $r_{q,2}=\eta^q(2\log \eta^{-q})^{(2-q)/2}$ for $0<q<2$. When $q=2$ and $(\log(p))^{3/2}/\sqrt{p}<\eta^2$, the SURE-based estimator attains the minimax risk. However,
when $(\log(p))^{3/2}/\sqrt{p}$ is of larger order compared to $\eta^2$, the error term dominates. In this case, we may use MMLE to tune $(w,c)$. When $c=0$ and $w$ is estimated by the MMLE, JS (2004) showed that the posterior median and the posterior mean are both minimax optimal for $q\in (1,2]$.
Therefore, one may combine SURE and empirical Bayes in a way similar to \citet{DonohoSURE}, depending on the sparsity of the signals.
}
\end{remark}





\section{Heteroscedastic models}\label{sec:hete}


In this section, we extend our
results to the heteroscedastic case (i.e., the unequal variance case). To this end, consider the
model,
\begin{align*}
X_i=\mu_i+\epsilon_i,\quad \epsilon_i\sim^{i.i.d} N(0,\sigma_i^2),
\end{align*}
for $1\leq i\leq p.$ As before, we impose the mixture prior distribution on $\mu_i$ i.e.,
$f(\mu)=(1-w)\delta_0(\mu)+w\gamma(\mu;b,c),$ where the nonzero
component of the prior, $\gamma$, belongs to a location-scale
family. Recall that $g(x;b,c)$ denotes the convolution between
$\phi(\cdot)$ and $\gamma(\cdot;b,c)$. Direct calculation shows that
$\int_{-\infty}^{+\infty}(1/\sigma_i)\phi((x-\mu)/\sigma_i)\gamma(\mu;b,c)d\mu=(1/\sigma_i)g(x/\sigma_i,b\sigma_i,c/\sigma_i).$
The MMLE $(\hat{w},\hat{b},\hat{c})$ is then defined as,
\begin{align}\label{margin}
(\hat{w},\hat{b},\hat{c})=\arg\max
\sum^{n}_{i=1}\log\{(1-w)\phi(Y_i)+wg(Y_i;b\sigma_i,c/\sigma_i)\},\quad
Y_i:=X_i/\sigma_i,
\end{align}
subject to the constraints that $b>0$ and $0\leq w\leq 1$.


We propose an alternative method below that takes into account the order information in the variances, which is useful in estimating the means [see \citet{xie2012sure}].
From (\ref{margin}), we see that $b_i:=b\sigma_i$ is a monotonic increasing function of $\sigma_i$. 
In other words, we have $b_i\geq b_j$ if $\sigma_i\geq \sigma_j\geq
0$. This observation suggests us to consider
the optimization problem,
\begin{align}\label{margin2}
(\hat{w},\hat{b}_1,\dots,\hat{b}_p,\hat{c})=\arg\max
\sum^{n}_{i=1}\log\{(1-w)\phi(Y_i)+wg(Y_i;b_i,c/\sigma_i)\},
\end{align}
subject to the ordering constraint
\begin{align}\label{mon}
b_i\geq b_j>0 \quad \text{if} \quad \sigma_i\geq \sigma_j.
\end{align}
Here we impose a monotone constraint on $\{b_i\}$ according to the
ordering of the variances.
We shall call the resulting estimator semi-parametric MMLE.
As seen in Section \ref{sec:sim}, the performance of the normal
density component and double exponential density component are
generally close in the homogeneous case. Therefore, we shall focus on the case of normal
prior, and develop an efficient algorithm to solve (\ref{margin2}).
Our algorithm is a modification of the EM algorithm which invokes
the PAV algorithm in its M-step. The details are summarized in Algorithm 1 below.

\begin{algorithm}[!h]\label{alg11}\small
\caption{}
0. Input the initial values $(w^{(0)},c^{(0)},b_1^{(0)},\dots,b_p^{(0)})$.\\
1. \textbf{E-step:} Given $(w,c,b_1,\dots,b_p)$, let
  $$Q_{1i}=\frac{(1-w)\phi(Y_i)}{(1-w)\phi(Y_i)+wg(Y_i;\tau_i^{-1/2},c/\sigma_i)} \quad \text{and} \quad
  Q_{2i}=1-Q_{1i},$$
  where $\tau_i=1/b_i^2$ for $1\leq i\leq p.$\\
2. \textbf{M-step:} For fixed $c$, solve the optimization problem
\begin{align}\label{opt-mon}
(\hat{\tau}_1,\dots,\hat{\tau}_p)=\arg\min\sum^{p}_{i=1}Q_{2i}\left\{\log(1+\tau_i)+\frac{(Y_i-c/\sigma_i)^2}{1+\tau_i}\right\}\quad
\text{subject to}\quad 0\leq \tau_i\leq \tau_j \quad \text{if} \quad
\sigma_i\geq \sigma_j,
\end{align}
For fixed $(\tau_1,\dots,\tau_p)$, let
\begin{align}\label{opt-cw}
\hat{c}=\frac{\sum_{i=1}^{p}Q_{2i}Y_i/\{\sigma_i(1+\tau_i)\}}{\sum_{i=1}^{p}Q_{2i}/\{\sigma_i^2(1+\tau_i)\}}\quad
\text{and}\quad
\hat{w}=\frac{1}{p}\sum^{p}_{i=1}Q_{2i}.
\end{align}
Iterate between (\ref{opt-mon}) and (\ref{opt-cw}) until
convergence.\\
3. Repeat the above E-step and M-step until the algorithm converges.
\end{algorithm}

Define
  \begin{align*}
  l(w,\tau_1,\dots,\tau_p,c)=&\sum^{p}_{i=1}Q_{1i}\left\{\log(1-w)-\log(Q_{1i})-\frac{Y_i^2}{2}\right\}
  \\&+\sum^{p}_{i=1}Q_{2i}\left\{\log(w)-\log(Q_{2i})-\frac{1}{2}\log(1+\tau_i)-\frac{(Y_i-c/\sigma_i)^2}{2+2\tau_i}\right\}.
  \end{align*}
Consider the optimization problem,
\begin{align}\label{opt-main}
\max_{w,\tau_1,\dots,\tau_p,c} l(w,\tau_1,\dots,\tau_p,c)  \quad
\text{subject to}\quad 0\leq  \tau_i\leq \tau_j \quad \text{if} \quad
\sigma_i\geq \sigma_j.
\end{align}
For fixed $c$, maximizing $l(w,\tau_1,\dots,\tau_p,c)$ with respect
to $(\tau_1,\dots,\tau_p)$ is equivalent to solving (\ref{opt-mon}).
On the other hand, for fixed $(\tau_1,\dots,\tau_p)$, the maximizers
of $l(w,\tau_1,\dots,\tau_p,c)$ with respect to $w$ and $c$ are
given in (\ref{opt-cw}). Therefore, the iteration between
(\ref{opt-mon}) and (\ref{opt-cw}) is essentially a coordinate
descent algorithm for solving (\ref{opt-main}).

The order constraint optimization problem (\ref{opt-mon}) can be
solved effectively using the PAV algorithm for isotonic regression.
Notice that
$$(Y_i-c/\sigma_i)^2-1=\arg\min_{\tau_i}\{\log(1+\tau_i)+(Y_i-c/\sigma_i)^2/(1+\tau_i)\}.$$
Consider the weighted isotonic regression,
\begin{align}
(\tilde{\tau}_1,\dots,\tilde{\tau}_p)=\arg\min\sum^{p}_{i=1}Q_{2i}\left\{(Y_i-c/\sigma_i)^2-1-\tau_i\right\}^2
\quad \text{subject to}\quad 0\leq \tau_i\leq \tau_j \quad \text{if}
\quad \sigma_i\geq \sigma_j.
\end{align}
Let $\hat{\tau}_i=\max\{\tilde{\tau}_i,0\}$ for $1\leq i\leq p.$ By
Chapter 1 of \citet{robertsonorder},
we have the following result.

\begin{proposition}\label{prop-iso}
The solution to (\ref{opt-mon}) is
$(\hat{\tau}_1,\dots,\hat{\tau}_p)'$.
\end{proposition}

\begin{remark}
{\rm Notice that $c/\sigma_i$ is a monotonic increasing function of
$\sigma_i$ if $c<0$ while it is monotonic decreasing when $c>0.$
However, as the sign of $c$ is generally unknown, it seems less
convenient to use the monotonic constraint on location parameters. }
\end{remark}

To end this subsection, we remark that the method can also be
extended to the mixture models described in Section \ref{sec:mix}.
In particular, one can consider the following MMLE,
\begin{align*}
(\hat{w}_0,\hat{w}_k,\hat{b}_{ki},\hat{c}_k)_{k=1,2,\dots,d}=\arg\max\sum_{i=1}^{n}\log\left\{(1-w_0)\phi(Y_i)+\sum_{k=1}^{d}w_{k}g(Y_i;b_{ki},c_k/\sigma_i)\right\},
\end{align*}
subject to the ordering constraint
\begin{align}\label{mon}
b_{ki}\geq b_{kj}>0 \quad \text{if} \quad \sigma_i\geq \sigma_j,
\end{align}
and $\sum_{k=0}^{d}w_k=1$ for $w_k\geq 0.$ The EM + PAV algorithm
can again be employed to solve the optimization problem. The details
of the algorithm are presented in Section \ref{sec:empav}.

\section{Numerical studies}\label{sec:sim}
\subsection{Two component mixture priors}\label{sec:sparse2}
We conduct simulation studies to compare and contrast the method in
Section \ref{sec:sparse1} with JS (2004) as well as the general
maximum likelihood empirical Bayes (denoted by GMLEB and S-GMLEB) in
\citet{jiang2009general}, shape constrained rule (SCR) in \citet{koenker2014convex} and the nonparametric empirical Bayes method (NEB) in
\citet{brown2009nonparametric}. We consider two prior density
components namely the double exponential and normal densities. Following the
well-established design of JS (2004), we generate a single
observation $\mathbf{X}\sim N(\mu_0,I_p)$ with $p=1000.$ Here
$\mu_0$ contains $k=5,50$ or $500$ nonzero entries with the same
value $v=3,4,5$ or $7$.

The simulation results are summarized in Table \ref{tb1}. Because
the non-null observations are being shrunk toward the data-driven
location, the proposed method outperforms JS (2004) and the nonparametric competitors in all cases as the nonzero entries are all equal.
The posterior median has slightly higher squared errors comparing to the posterior mean.
However, it produces an exact sparse solution, which is desirable if the goal
is to recover the support of signals or do feature selection. We also note that the two density components perform similarly despite their different tail behaviors.

Table \ref{tb2} reports the MMLE for $w$ as well as the false positive numbers
(FP) and false negative numbers (FN) for the posterior median. The FP for our method
is consistently lower than that of JS (2004). As the underlying model is indeed a two-component normal mixture, $\hat{w}$ in our method provides
a reasonable estimation of the nonzero proportion when the signal
strength is relatively strong or the signal is not too sparse.
However, when the location parameter $c$ is set to be zero in JS
(2004), $\hat{w}$ provides a less meaningful estimation of the
nonzero proportion. Furthermore, Table \ref{tb11} summarizes the
average of total $\ell_1$ loss for the proposed method, JS (2004)'s
approach as well as the posterior mean and posterior median based on
\citet{kiefer1956consistency}'s nonparametric maximum likelihood
estimator (NPMLE). We implement Kiefer and Wolfowitz's procedure
using the R package \texttt{REBayes}; see \citet{koenker2016package}. It
is clear that the proposed method outperforms other approaches in
this case. Although the Bayes rule (posterior mean) based on NPMLE
has superior performance in terms of $l_2$ loss, its $l_1$ loss is
considerably higher which is likely due to the non-sparseness of its
solution.

In Table \ref{tb3}, we further report some simulation results
following the setting in Table 4 of \citet{jiang2009general}, where
$p=1000$ and $\mu_j\sim^{i.i.d} N(\tilde{\mu},\sigma^2)$. For such
design, James-Stein estimator is the best performer. It is
interesting to see that our
method performs as well as the James-Stein estimator when normal density is employed.
Note that in this setup, the performance of the posterior median in JS (2004) considerably worsens
and the improvement by including a location parameter is
significant.

We also note that the posterior mean based on SURE performs competitively with the empirical Bayes counterpart.
Overall, our method has reasonably good finite sample performance at
the expense of low computational overhead compared to nonparametric empirical Bayes and having the advantage of no tuning as
compared to the nonparametric approaches.

\begin{table}[h]\footnotesize
\caption{\small Average of total squared error of estimation of
various methods on a mixed signal of length 1000. The numbers for
GMLEB, S-GMLEB, SCR and NEB are adapted from \citet{jiang2009general},
\citet{koenker2014convex}, and \citet{brown2009nonparametric}
respectively. The results for SCR and NEB are based on 1000 and 50
simulation runs, while the results for other methods are based on
100 simulation runs. Boldface entries denote the best
performer.}\label{tb1}
\begin{center}
\begin{tabular}{l rrrr rrrr rrrr}\toprule
&\multicolumn{4}{c}{$k=5$}&\multicolumn{4}{c}{$k=50$}&\multicolumn{4}{c}{$k=500$}
\\\cmidrule(r){2-5}\cmidrule(r){6-9}\cmidrule(r){10-13} $v$ &
 3 & 4 & 5 & 7&  3 & 4 & 5 & 7 &  3 & 4 & 5 & 7  \\\midrule
L-Exp (median) & 34  & 26 & 16  & 4  & 178 & 117 & 53 & 7 & 551 & 341 & 141 & \textbf{9}\\
L-Exp (mean)  &  \textbf{32}  & \textbf{25} & \textbf{15} & 4 & \textbf{148} & \textbf{97} & 46 & 8 & 445 & \textbf{277} & \textbf{119} & 13 \\
L-Normal (median) & 35 & 28 & 17 & \textbf{3} & 184 & 123 & 53 & \textbf{5} & 584 & 366 & 160 & 14 \\
L-Normal (mean) & 34 & 27 & \textbf{15} & \textbf{3} & 155 & 102 & \textbf{45} & \textbf{5} & \textbf{443} & 279 & 124 & 12\\

L-Exp-S (mean)   & 35 & 28 & 16 & 5 & 153 & 102 & 46 & 7 & 447 & 283 & 129 & 19 \\
L-Normal-S (mean)& 35 & 28 & \textbf{15} & 5 & 153 & 102 & 46 & 7 & 444 & 280 & 126 & 16 \\


Exp  & 36  & 30  & 18  & 9 & 211 & 151 & 101 & 72 & 852 & 870 & 780 & 656 \\
GMLEB  & 39 & 34 & 23 & 11  & 157 & 105 & 58 & 14 & 459 & 285 & 139 & 18\\
S-GMLEB & \textbf{32} & 28 & 17 & 6  & 150 & 99 & 54 & 10  & 454 & 282 & 136 & 15\\
SCR & 37 & 34 & 21 & 11 & 173 & 121 & 63 & 16 & 488 & 310 & 145 &
22\\
 NEB  & 53 & 49 & 42 & 27 &
179 & 136 & 81 & 40 & 484 & 302 & 158& 48
\\ \midrule
\end{tabular}
\\Note: L-Exp/L-Normal (L-Exp-S/L-Normal-S) denote the proposed empirical Bayes (Stein's) method, where the density component of the prior is double exponential or normal with location shift.
\end{center}
\end{table}

\begin{table}[h]\footnotesize
\caption{\small MMLE for $w$, and the
false positive numbers (FP) and false negative numbers (FN) for the posterior medians of
the proposed method and JS (2004)'s method.}\label{tb2}
\begin{center}
\begin{tabular}{ll rrrr rrrr rrrr}\toprule
&
&\multicolumn{4}{c}{$k/p=0.005$}&\multicolumn{4}{c}{$k/p=0.05$}&\multicolumn{4}{c}{$k/p=0.5$}
\\\cmidrule(r){3-6}\cmidrule(r){7-10}\cmidrule(r){11-14}  $v$ &
& 3 & 4 & 5 & 7&  3 & 4 & 5 & 7 &  3 & 4 & 5 & 7  \\\midrule
L-Exp & $\hat{w}$   & 0.086 & 0.022 & 0.010 &0.005 & 0.05 & 0.05 & 0.05 &0.05 & 0.51 & 0.50 & 0.50 & 0.50\\
& FP & 15.8 & 2.0 & 0.7 & 0.1 & 6.5  & 3.2 & 0.9 & 0.0 & 36.2 & 12.2 & 3.3 & 0.1\\
& FN & 2.7  & 0.9 & 0.3 & 0.0 & 14.5 & 4.6 & 1.3 & 0.0 & 30.6 & 10.6 & 2.8 & 0.1\\
L-Normal & $\hat{w}$   & 0.051 & 0.017 & 0.008 &0.006 & 0.08 & 0.05 & 0.05 &0.05 & 0.50 & 0.50 & 0.50 & 0.50\\
                       & FP & 1.2 & 0.9 & 0.5 & 0.1 & 8.1  & 3.6 & 1.0 & 0.1 & 35.1 & 12.0 & 3.1 & 0.2\\
                       & FN & 3.2  & 1.2 & 0.3 & 0.0 & 14.5 & 4.7 & 1.2 & 0.0 & 31.6 & 11.3 & 3.4 & 0.1\\


Exp & $\hat{w}$ & 0.137 & 0.056 & 0.029 & 0.014 & 0.23 & 0.18 & 0.14 & 0.10 & 1.00 & 1.00 & 0.89 & 0.74\\
& FP & 33.2 & 11.0 & 0.9 & 0.5 & 15.3 & 10.8 & 7.4 & 3.8 & 500 & 500 & 310.0 & 97.3\\
& FN & 3.2  & 1.2  & 0.3 & 0.0 & 14.5 & 3.5  & 0.6 & 0.0 & 0.0 & 0.0 & 0.0 & 0.0\\
\midrule
\end{tabular}
\\Note: L-Exp/L-Normal denote the proposed empirical Bayes method, where the density component of the prior is double exponential or normal with location shift.
\end{center}
\end{table}

\begin{table}[h]\footnotesize
\caption{\small Average of total $\ell_1$ loss of estimation of various
methods on a mixed signal of length 1000. The results are based on
100 simulation runs. Boldface entries denote the best
performer.}\label{tb11}
\begin{center}
\begin{tabular}{l rrrr rrrr rrrr}\toprule
&\multicolumn{4}{c}{$k=5$}&\multicolumn{4}{c}{$k=50$}&\multicolumn{4}{c}{$k=500$}
\\\cmidrule(r){2-5}\cmidrule(r){6-9}\cmidrule(r){10-13} $v$ &
 3 & 4 & 5 & 7&  3 & 4 & 5 & 7 &  3 & 4 & 5 & 7  \\\midrule
L-Exp (median) & \textbf{13}  & \textbf{9} & \textbf{6} & \textbf{3}  & \textbf{69} & \textbf{38} & \textbf{18} & \textbf{8} &  229 & 119 & 56 & 27\\
L-Exp (mean)   & 38 &  22 & 11 & 3 & 107 & 57 & 25 & 9 & 327 & 173 & 83 & 38 \\
L-Normal (median) & 21 & 10 & \textbf{6} & \textbf{3} & 72 & 39 & \textbf{18} & \textbf{8}  & \textbf{225} &\textbf{116} & \textbf{52} & \textbf{ 23 }\\
L-Normal (mean)   & 39 & 23 & 12 & 4 & 115 & 58 & 25 & 8 & 310 & 157 & 68 & 24 \\

L-Exp-S (mean)   & 32   & 19 & 8 & 4 & 113 & 60 & 26 & 9 & 329 & 177 & 90 & 44\\
L-Normal-S (mean)& 33 &  19 & 8 & 4 & 114 & 59 & 25 & 8 & 312 & 158 & 70 & 24 \\


Exp  & 14     & 11 & 8 & 6 & 96 & 73 & 58 & 49 & 708 & 720 & 620 & 501 \\
NPMLE (median)& 66 & 65 & 64 & 64  & 125 & 93 & 79 & 71 & 274 & 164 & 99 & 72\\
NPMLE (mean)  & 52 & 47 & 41 & 37  & 134 & 88 & 59 & 43 & 329 & 181
& 95 & 51
\\ \midrule
\end{tabular}
\\Note: L-Exp/L-Normal (L-Exp-S/L-Normal-S) denote the proposed empirical Bayes (Stein's) method, where the density component of the prior is double exponential or normal with location shift.
\end{center}
\end{table}

\begin{table}[h]\footnotesize
\caption{\small Average of total squared error of estimation of
various methods on a mixed signal of length 1000. The numbers for
James-Stein, GMLEB, and S-GMLEB are adapted from \citet{jiang2009general}. Boldface entries denote the best two performers.}\label{tb3}
\begin{center}
\begin{tabular}{l rrrr rrrr rr}\toprule
&\multicolumn{4}{c}{$\sigma^2=0.1$}&\multicolumn{3}{c}{$\sigma^2=2$}&\multicolumn{3}{c}{$\sigma^2=40$}
\\\cmidrule(r){2-5}\cmidrule(r){6-8}\cmidrule(r){9-11} $\tilde{\mu}$ &
 3 & 4 & 5 & 7&  3  & 5 & 7 &  3 &  5 & 7  \\\midrule
L-Exp (median) & 94 & 94 & 93 & \textbf{92} & 722 & 704 & 704 & 989 & 1007 & 1014 \\
L-Exp (mean) &  93 & 93 & 93 & 93 & 692  & 689 & 689 & 986 & 990 & 994\\
L-Normal (median) & 94 & 95 & 94 & 93 & 667 & 666  & 666 & 978  & \textbf{977} & 977 \\
L-Normal (mean) &  94 & 94 & 93 & 93 & 666 &\textbf{ 665} & 666 & 977 & \textbf{977} & 977 \\

L-Exp-S (mean)   & \textbf{92}  & \textbf{92} & \textbf{92} & \textbf{92} & 685 & 684 & 684 & 982 & 982 & 982 \\
L-Normal-S (mean) & 94 & 93 & 93 & 93 & \textbf{665} & \textbf{664} & \textbf{664} & \textbf{974} & \textbf{974} & \textbf{974} \\


Exp   & 1086 & 1066 & 1044 & 1022 & 1020 & 1037 & 1022 & 990 & 994 & 999\\
GMLEB  & 94 & 94 & 95 & 95 & 675 & 678 & 673 & 1001 & 1015 & 1009   \\
S-GMLEB  & 97 & 98  & 99 & 98 & 678 & 681 & 675 & 1002 & 1015 & 1009 \\
James-Stein & \textbf{92} & \textbf{92} & \textbf{92} & 93
& \textbf{665} & 670 & \textbf{665} & \textbf{970} &
982 & \textbf{975}
\\ \midrule
\end{tabular}
\\Note: L-Exp/L-Normal (L-Exp-S/L-Normal-S) denote the proposed empirical Bayes (Stein's) method, where the density component of the prior is double exponential or normal with location shift.
\end{center}
\end{table}

\subsection{Finite mixture priors}
To evaluate the performance of the method proposed in Section
\ref{sec:mix}, we modify the setting in JS (2004) by considering the
models with $\mu_i=v$ for $1\leq i\leq k$ and $\mu_i=-v$ for
$k+1\leq i\leq 2k,$ where $v=3,4,5,7$ and $k=5,50,250.$ To conserve
space, we only present the results with normal density components.
As seen from Table \ref{tb-mix}, when $m\geq 2$, the posterior mean
and median based on the finite mixture models perform as well as their NPMLE counterparts.
For $k=50$ and $k=250$,
we see a significant improvement by including additional mixing
component(s). The total square errors are not sensitive to the choice
of $m$ as long as $m\geq 2.$ Table \ref{tb-mix2} summarizes the
false positive/negative numbers (FP/FN) for the posterior median.
The mixture models with $m\geq 2$ greatly reduce the FP numbers for
$k=50,250.$ However, the over-fitted models may deliver higher
false positive numbers for dense and weak signals as compared to the
correctly specified model. To select the number of components, we
implement the BIC criterion described in (\ref{bic}) with the upper
bound $M_0=6.$ It is seen that the BIC criterion generally selects the true number of clusters and
the corresponding estimators perform reasonably well when the signals are not too weak or sparse.

\begin{table}[h]\footnotesize
\caption{\small Average of total squared error of estimation of
various methods on a mixed signal of length 1000. The results are
based on 100 simulation runs. }\label{tb-mix}
\begin{center}
\begin{tabular}{ll rrrr rrrr rrrr}\toprule
&
&\multicolumn{4}{c}{$2k=10$}&\multicolumn{4}{c}{$2k=100$}&\multicolumn{4}{c}{$2k=500$}
\\\cmidrule(r){3-6}\cmidrule(r){7-10}\cmidrule(r){11-14}  & $m$ &
 3 & 4 & 5 & 7&  3 & 4 & 5 & 7 &  3 & 4 & 5 & 7  \\\midrule
L-Normal (mean)& 1 & 61 & 51 & 32 & 18 & 320 & 264 & 198 & 150 & 821 & 821 & 748 & 663\\
L-Normal (median)  & 1 & 65 & 53 & 28 & 16 & 334 & 240 & 168 & 133 & 821 & 779 & 693 & 618 \\
L-Normal (mean)& 2 & 62 & 53 & 33 & 18 & 300 & 203 & 94  & 10  & 628 & 391 & 168 & 14\\
L-Normal (median)  & 2 & 65 & 54 & 28 & 14 & 370 & 246 & 114 & 12  & 803 & 505 & 213 & 17 \\
L-Normal (mean)& 3 & 63 & 53 & 33 & 19 & 301 & 204 & 97  & 14  & 630 & 394 & 172 & 19 \\
L-Normal (median)  & 3 & 65 & 53 & 28 & 15 & 370 & 244 & 112 & 13  & 792 & 499 & 209 & 18\\
L-Normal (mean)& 4 & 63 & 53 & 34 & 19 & 301 & 204 & 97 & 14   & 631 & 395 & 173 & 19\\
L-Normal (median)  & 4 & 65 & 53 & 28 & 15 & 371 & 244 & 11 & 14   & 790 & 498 & 211 & 19\\
L-Normal (mean)& 5 & 63 & 53 & 34 & 19 & 301 & 205 & 97 & 15   & 631 & 396 & 173 & 20\\
L-Normal (median)  & 5 & 65 & 53 & 28 & 15 & 371 & 244 & 111 & 14  & 793 & 498 & 210 & 19\\
L-Normal (mean)  & BIC & 61 & 51 & 32 & 18 & 318 & 205 & 94 & 10 & 628 & 391 & 168 & 14 \\
L-Normal (median)& BIC & 65 & 53 & 28 & 16 & 338 & 245 & 114 & 12 & 803 & 505 & 213 & 17 \\
Exp (median) & NA & 64 & 52 & 28 & 16 & 335 & 251 & 180 & 140 & 860 & 875 & 786 & 659 \\
NPMLE (mean)   & NA& 63 & 53 & 30 & 10  & 302 & 206 & 99 & 16 & 633 & 397 & 174 & 22\\
NPMLE (median)     & NA& 74 & 63 & 39 & 17  & 383 & 255 & 120 & 27&
830 & 516 & 221 & 34 \\ \midrule
\end{tabular}
\end{center}
\end{table}

\begin{table}[h]\footnotesize
\caption{\small The false positive numbers (FP) and false negative
numbers (FN) for the posterior median based on the finite mixture
models.}\label{tb-mix2}
\begin{center}
\begin{tabular}{ll rrrr rrrr rrrr}\toprule
&
&\multicolumn{4}{c}{$2k/p=0.01$}&\multicolumn{4}{c}{$2k/p=0.10$}&\multicolumn{4}{c}{$2k/p=0.5$}
\\\cmidrule(r){3-6}\cmidrule(r){7-10}\cmidrule(r){11-14} & $m$
& 3 & 4 & 5 & 7&  3 & 4 & 5 & 7 &  3 & 4 & 5 & 7  \\\midrule
 FP & 1& 1.9 & 1.6 & 1.3 & 0.8 & 28.5 & 20.5 & 13.7 & 7.1 & 495.6 & 307.6 & 153.2 & 61.5 \\
 FN & 1& 6.0 & 2.3 & 0.4 & 0.0 & 20.6 & 4.5  & 0.5  & 0.0 & 0.0  & 0.0 & 0.0 & 0.0\\
 FP & 2& 2.3 & 2.1 & 1.6 & 0.9 & 15.0 & 7.0 & 2.3 & 0.2   & 45.8 & 16.7 & 4.6 & 0.2 \\
 FN & 2& 5.9 & 2.1 & 0.3 & 0.0 & 28.7 & 9.3 & 2.5 & 0.0   & 48.2 & 16.1 & 4.1 & 0.1\\
 FP & 3& 2.9 & 2.7 & 2.0 & 1.1 & 16.1 & 7.8 & 2.9 & 0.5   & 98.5 & 21.6 & 5.8 & 0.6\\
 FN & 3& 5.8 & 2.0 & 0.3 & 0.0 & 28.0 & 8.8 & 2.2 & 0.0   & 31.2 & 13.6 & 3.6 & 0.0\\
 FP & 4& 4.0 & 3.7 & 2.7 & 1.4 & 16.7 & 8.0 & 3.0 & 0.6   & 203.9 & 23.0 & 6.0 & 0.7\\
 FN & 4& 5.4 & 1.8 & 0.2 & 0.0 & 27.6 & 8.7 & 2.1 & 0.0   & 23.8 & 13.0 & 3.4 & 0.0\\
 FP & 5& 7.7 & 6.8 & 4.5 & 2.2 & 17.2 & 8.1 & 3.0 & 0.6   & 280.6 & 23.8 & 6.1 & 0.7 \\
 FN & 5& 4.9 & 1.7 & 0.2 & 0.0 & 27.4 & 8.6 & 2.1 & 0.0   & 19.5 & 12.8 & 3.4 & 0.0 \\
 FP & BIC & 1.9 & 1.6 & 1.3 & 0.8 & 27.5 & 7.7 & 2.3 & 0.2& 45.8 & 16.6 & 4.6 & 0.2\\
 FN & BIC & 6.0 & 2.3 & 0.4 & 0.0 & 21.4 & 9.1 & 2.5 & 0.1& 48.2 & 16.1 & 4.1 & 0.1\\
 FP (Exp) & NA & 2.7 & 1.7 & 1.4 & 0.8 & 51.9 & 28.9 & 18.0 & 8.9 & 500.0 & 500.0 & 317.5 & 98.6\\
 FN (Exp) & NA & 5.8 & 2.1 & 0.3 & 0.0 & 14.3 & 3.3 & 0.4 & 0.0 & 0.0 & 0.0 & 0.0 & 0.0\\
\midrule
\end{tabular}
\end{center}
\end{table}

\subsection{Heteroscedastic models}
In this subsection, we present some numerical results to demonstrate
the finite sample performance of the semi-parametric MMLE for
heteroscedastic models. To this end, we generate a single
observation $\mathbf{X}\sim N(\mu_0,\Sigma)$, where
$\mu_0=(\mu_1,\dots,\mu_p)$ and
$\Sigma=\text{diag}(\sigma_1^2,\dots,\sigma_p^2).$ Consider the
following models, where $v=3,5,4,7$, and $K=5,50,500$.
\begin{enumerate}
  \item[(A)] $\sigma_i\sim^{i.i.d} \text{Unif}(1,1.5)$. Sort $\{\sigma_i\}$ so that $\sigma_1\leq \sigma_2\leq \cdots\leq \sigma_p$. Let $\mu_j=v$ for $1\leq j\leq K$ and zero otherwise.
  \item[(B)] $\sigma_i\sim^{i.i.d} \text{Unif}(1,1.5)$. Sort $\{\sigma_i\}$ so that $\sigma_1\geq \sigma_2\geq \cdots\geq \sigma_p$. Let $\mu_j=v$ for $1\leq j\leq K$ and zero otherwise.
  \item[(C)] $\sigma_i\sim^{i.i.d} \text{Unif}(1,1.5)$. Sort $\{\sigma_i\}$ so that $\sigma_1\leq \sigma_2\leq \cdots\leq \sigma_p$. Let $\mu_j\sim N(v,1)$ for $1\leq j\leq K$ and zero otherwise.
  \item[(D)] $\sigma_i\sim^{i.i.d} \text{Unif}(1,1.5)$. Sort $\{\sigma_i\}$ so that $\sigma_1\geq \sigma_2\geq \cdots\geq \sigma_p$. Let $\mu_j\sim N(v,1)$ for $1\leq j\leq K$ and zero otherwise.
  \item[(E)] $\sigma_i\sim^{i.i.d} \text{Unif}(1,1.5)$. Sort $\{\sigma_i\}$ so that $\sigma_1\leq \sigma_2\leq \cdots\leq \sigma_p$. Let $\mu_j=v$ for $\lfloor(p-K)/2\rfloor \leq j\leq \lfloor(p+K)/2\rfloor-1$ and zero otherwise.
  \item[(F)] $\sigma_i\sim^{i.i.d} \text{Unif}(1,1.5)$. Sort $\{\sigma_i\}$ so that $\sigma_1\leq \sigma_2\leq \cdots\leq \sigma_p$. Let $\mu_j\sim N(v,1)$ for $\lfloor(p-K)/2\rfloor \leq j\leq \lfloor(p+K)/2\rfloor-1$ and zero otherwise.
  \item[(G)] $\sigma_i\sim^{i.i.d} \text{Unif}(1,1.5)$. Sort $\{\sigma_i\}$ so that $\sigma_1\geq \sigma_2\geq \cdots\geq \sigma_p$. Let $\mu_j=v$ for $\lfloor(p-K)/2\rfloor \leq j\leq \lfloor(p+K)/2\rfloor-1$ and zero otherwise.
  \item[(H)] $\sigma_i\sim^{i.i.d} \text{Unif}(1,1.5)$. Sort $\{\sigma_i\}$ so that $\sigma_1\geq \sigma_2\geq \cdots\geq \sigma_p$. Let $\mu_j\sim N(v,1)$ for $\lfloor(p-K)/2\rfloor \leq j\leq \lfloor(p+K)/2\rfloor-1$ and zero otherwise.
  \item[(I)] $\sigma_i\sim^{i.i.d} \text{Unif}(1,1.01)$. Let $\mu_j=v$ for $1\leq j\leq K$ and zero otherwise.
  \item[(J)] $\sigma_i\sim^{i.i.d} \text{Unif}(1,1.01)$.  Let $\mu_j\sim N(v,1)$ for $1\leq j\leq K$ and zero otherwise.
\end{enumerate}

\begin{table}[h]\footnotesize
\caption{\small Average of total squared error of estimation of
various methods on a mixed signal of length 1000. The results are based on
100 simulation runs.}\label{tb-het1}
\begin{center}
\begin{tabular}{ll rrrr rrrr rrrr}\toprule
& &\multicolumn{4}{c}{$k=5$}&\multicolumn{4}{c}{$k=50$}&\multicolumn{4}{c}{$k=500$}
\\\cmidrule(r){3-6}\cmidrule(r){7-10}\cmidrule(r){11-14} & &
 3 & 4 & 5 & 7&  3 & 4 & 5 & 7 &  3 & 4 & 5 & 7  \\\midrule
 (A) & L-Normal (median) & 33  & 28 & 21 & 10 & 188 & 166 & 116 & 30 & 1027 & 901 & 619 & 164\\
 & L-Normal (mean) & 33 & 29 & 21 & 9 & 174 & 148 & 98 & 26 & 788 & 690 & 480 & 130 \\
 & Semi  (median)& 12 & 11 & 11  & 11 & 56 &  57 & 57  & 58 & 947  & 836 & 598 & 169 \\
  & Semi (mean)  & 12 & 12 & 12  & 12 & 56 &  57 & 57  & 58 & 773  & 689 & 484 & 134 \\
 (B) & L-Normal (median)    & 49  & 76 & 83 & 39 & 375 & 428 & 364 & 123 & 1020 & 908 & 613 & 192\\
  & L-Normal (mean)& 48  & 72 & 75 & 34 & 303 & 334 & 276 & 99  & 772  & 668 & 460 & 135 \\
 & Semi (median)            & 48  & 71 & 69 & 39 & 356 & 409 & 334 & 109 & 678 & 580 & 445 & 209\\
   & Semi (mean)   & 49  & 69 & 66 & 37 & 295 & 345 & 289 & 99  & 542 & 505  & 412 & 220\\
 (C) & L-Normal (median)   & 26 & 27 & 24 & 12 & 170 & 154 & 128 & 59 & 949 & 965 & 833 & 492\\
  & L-Normal (mean)& 27 & 27 & 23 & 11 & 164 & 154 & 121 & 58 & 883 & 884 & 766 & 465\\
 & Semi (median)           & 11 & 11 & 11 & 10 & 56  & 57 & 58 & 58 & 875 & 889 & 787 & 497  \\
   & Semi (mean)   & 12 & 12 & 11 & 11 & 56  & 57 & 58 & 58 & 855 & 864 & 761 & 479  \\
 (D) & L-Normal (median)   & 50 & 72 & 74 & 51 & 359 & 407 & 358 & 187 & 1117 & 1122 & 964 & 581\\
  & L-Normal (mean)& 49 & 68 & 70 & 43 & 289 & 316 & 282 & 150 & 900 & 906 & 796 & 518 \\
 & Semi (median)           & 47 & 65 & 68 & 47 & 340 & 383 & 332 & 171 & 1075 & 1023 & 877 & 603\\
  & Semi (mean)    & 49 & 63 & 65 & 44 & 287 & 323 & 276 & 147 & 773 & 808 & 762 & 592\\
  (E) & L-Normal (median)    & 46 & 58 & 41 & 21 & 286 & 275 & 196 & 54 & 980 & 859 & 590 & 174 \\
   & L-Normal (mean)& 45 & 53 & 39 & 17 & 236 & 225 & 160 & 44 & 748 & 653 & 452 & 132\\
 & Semi (median)            & 43 & 51 & 37 & 20 & 251 & 230 & 179 & 126 & 895 & 799 & 587 & 200\\
   & Semi (mean)    & 43 & 50 & 41 & 25 & 223 & 233 & 205 & 152 & 677 & 630 & 473 & 181\\
  (F) & L-Normal (median)    & 43 & 48 & 42 & 22 & 255 & 260 & 208 & 97 & 1006 & 1025 & 886 & 527\\
   & L-Normal (mean)& 42 & 45 & 40 & 21 & 214 & 216 & 177 & 86 & 868 & 880 & 773 & 489\\
 & Semi (median)         & 39 & 43 & 38 & 24 & 223 & 220 & 186 & 127 & 1033 & 1000 & 864 & 555\\
   & Semi (mean) & 40 & 44 & 40 & 27 & 204 & 219 & 203 & 155 & 838 & 860 & 778 & 531 \\
  (G) & L-Normal (median)    & 46 & 53 & 40 & 16 & 286 & 276 & 193 & 53 & 979 & 858 & 591 & 175 \\
   & L-Normal (mean)& 45 & 49 & 36 & 15 & 236 & 224 & 158 & 44 & 747 & 653 & 452 & 133 \\
 & Semi (median)            & 42 & 47 & 34 & 20 & 250 & 229 & 178 & 126 & 891 & 798 & 587 & 200\\
   & Semi (mean)    & 43 & 46 & 38 & 24 & 223 & 233 & 204 & 152 & 676 & 629 & 473 & 182\\
  (H) & L-Normal (median)    & 41 & 47 & 41 & 22 & 254 & 253 & 209 & 96 & 1004 & 1010 & 887 & 530\\
   & L-Normal (mean)& 40 & 44 & 38 & 20 & 214 & 215 & 175 & 88 & 868  & 878 & 772 & 490\\
 & Semi (median)         & 37 & 42 & 37 & 22 & 224 & 221 & 183 & 129 & 1032 & 989 & 868 & 557\\
  & Semi (mean)  & 39 & 42 & 39 & 26 & 205 & 218 & 202 & 156 & 839 & 859 & 778 & 533\\
  (I) & L-Normal (median)   & 35 & 26 & 17 & 4 & 183 & 124 & 54 & 5 & 588 & 357 & 152 & 12\\
   & L-Normal (mean)& 33 & 26 & 16 & 4 & 152 & 102 & 45 & 5 & 448 & 276 & 116 & 9\\
 & Semi (median)            & 35 & 26 & 17 & 7 & 179 & 124 & 56 & 8 & 584 & 355 & 153 & 16 \\
   & Semi (mean)    & 35 & 30 & 21 & 9 & 160 & 109 & 49 & 8 & 452 & 281 & 121 & 13\\
  (J) & L-Normal (median)   & 29 & 26 & 21 & 7 & 166 & 144 & 97 & 43 & 670 & 598 & 476 & 296 \\
   & L-Normal (mean)& 27 & 24 & 20 & 7 & 141 & 122 & 86 & 40 & 593 & 537 & 433 & 288\\
 & Semi (median)            & 27 & 24 & 20 & 9 & 165 & 145 & 100 & 44 & 681 & 603 & 480 & 301\\
   & Semi (mean)    & 29 & 27 & 22 & 11& 145 & 126 & 88  & 42 & 598 & 542 & 438 & 293\\
\midrule
\end{tabular}
\end{center}
\end{table}

We compare the performance of the posterior mean and median delivered by the
MMLE in (\ref{margin}) and
(\ref{margin2}). The simulation results are reported in Table
\ref{tb-het1}. In models (A)-(D), the semi-parametric estimator generally
outperforms the estimator which dose not take into account the order
structure. We observe improvement regardless of the direction of the
order. In models (E)-(H) where the signals correspond to moderate
variances, the semiparametric approach delivers better results in
most cases when $v=3,4,5$. In models (I)-(J) which contain no order
information, the semiparametric procedure is very comparable with
the parametric procedure without using the order structure. Overall,
the performance of the semi-parametric approach is quite robust and
its computational cost is moderate due to the efficiency of the PAV algorithm.

\subsection{Application to wavelet approximation}
We apply the method in Section \ref{sec:hete} to wavelet coefficient
estimation. Suppose we have observations
$$X_i=h(t_i)+\epsilon_i$$
of a function $h(\cdot)$ at $N=2^J$ regularly spaced points $t_i$
with $\epsilon_i\sim N(0,\sigma^2_i)$, where $N$ and $J$ are
positive integers. Let $d_{jk}$ be the elements of the discrete
wavelet transformation (DWT) of the sequence $h(t_i)$. Similarly write
$d_{jk}^*$ the DWT of the observed data $X_i$. At the
$j$th level, we set up a model:
\begin{align}\label{w1}
d_{jk}^*=d_{jk}+\tilde{\sigma}_{jk}\varepsilon_{jk}, \quad
k=1,2,\dots,N_j,
\end{align}
where $\varepsilon_{jk}\sim N(0,1)$. At level $j$, we estimate
$d_{jk}$ by the posterior median
\begin{align*}
\hat{d}_{jk}=\delta^H(d_{jk}^*;\hat{w},\hat{b}_1,\dots,\hat{b}_p,\hat{c}),
\end{align*}
and the posterior mean,
\begin{align*}
\check{d}_{jk}=\zeta^H(d_{jk}^*;\hat{w},\hat{b}_1,\dots,\hat{b}_p,\hat{c}),
\end{align*}
where $(\hat{w},\hat{b}_1,\dots,\hat{b}_p,\hat{c})$ is the solution
to (\ref{margin2}) based on $\{d_{jk}^*\}_{k=1}^{N_j}$. In practice,
the noise $\tilde{\sigma}_{jk}$ are unknown and need to be replaced
by estimate $\hat{\sigma}_{jk}$. Finally, we
apply the inverse DWT to $\hat{d}_{jk}$ (or $\check{d}_{jk}$) to get the wavelet
approximation for $X_i$.


As an illustration, we employ the proposed method to process the
wavelet transform of a two-dimensional image. We consider the image
of Ingrid Daubechies contained in the \texttt{waveslim} package in R.
After loading the image, we reverse its sign, in order to obtain an
image that comes out in positive rather than negative when using the
image with the option \texttt{col=gray(1:100/100)} in R. We then
construct a noisy image by adding heteroscedastic normal noise to
each pixel. In particular, the standard deviation of the noise we
add to the $(i,j)$th pixels is $(i+j)/a_0$ for $a_0=10,15,20.$
Following \citet{silverman2005ebayesthresh}, we construct the
two-dimensional wavelet transform using the routine \texttt{dwt.2d}
and the Daubechies \texttt{d6} wavelet. As pointed out in \citet{silverman2005ebayesthresh},
it may be appropriate to use dictionaries
other than the standard two-dimensional wavelet transform. Here we
mainly use this example to illustrate how our method can be used in
a broader context. To estimate the standard deviation of the noise,
we partition the wavelet coefficients at the finest scale into
$m\times m$ blocks over space, and use median-absolute deviation to
estimate the standard deviation of noise at each of the $m^2$
blocks. In our analysis, we set $m=8$ and $16$, which deliver very
similar results. At each level, the wavelet coefficients in the same
block are assumed to have the same standard deviation. Figure
\ref{fig:dau} shows the original and noisy images. We apply the
method in Section \ref{sec:hete}, \citet{johnstone2005empirical}'s
procedure with the double exponential density component and the
NPMLE method (implemented in the R package \texttt{REBayes}) to the
wavelet coefficients at each level, and then invert the transform
using the R function \texttt{idwt.2d} to find the final estimate. To
implement \citet{johnstone2005empirical}'s approach, we let
$d_{ij}=\hat{\sigma}_{ij}\delta(d_{ij}/\hat{\sigma}_{ij};\hat{w},\hat{b})$
with $\hat{\sigma}_{ij}$ being the above blockwise estimate of the
standard deviation. Here $\delta(\cdot;\hat{w},\hat{b})$ denotes the
posterior median, and $(\hat{w},\hat{b})$ are the MMLEs with the location parameter being zero.

To quantify the performance of different methods, we consider
$\text{MSE}=\sum_{i,j=1}^{256}(h(t_{ij})-\hat{h}(t_{ij}))^2$, where
$h(t_{ij})$ and $\hat{h}(t_{ij})$ denote the $(i,j)$th pixel values
for the original image and the reconstructed image respectively.
Table \ref{tb-dau} summarizes the ratios of the MSE of the proposed
method and NPMLE procedure to that of \citet{johnstone2005empirical}. Both the semiparametric estimator and the NPMLE based
estimators provide an improvement over \citet{johnstone2005empirical}.
Our semiparametric approach is slightly better than the NPMLE in a few cases, and the posterior mean delivers better results as compared to the posterior median.

\begin{figure}[!h]
\centering
\includegraphics[height=6cm,width=5cm]{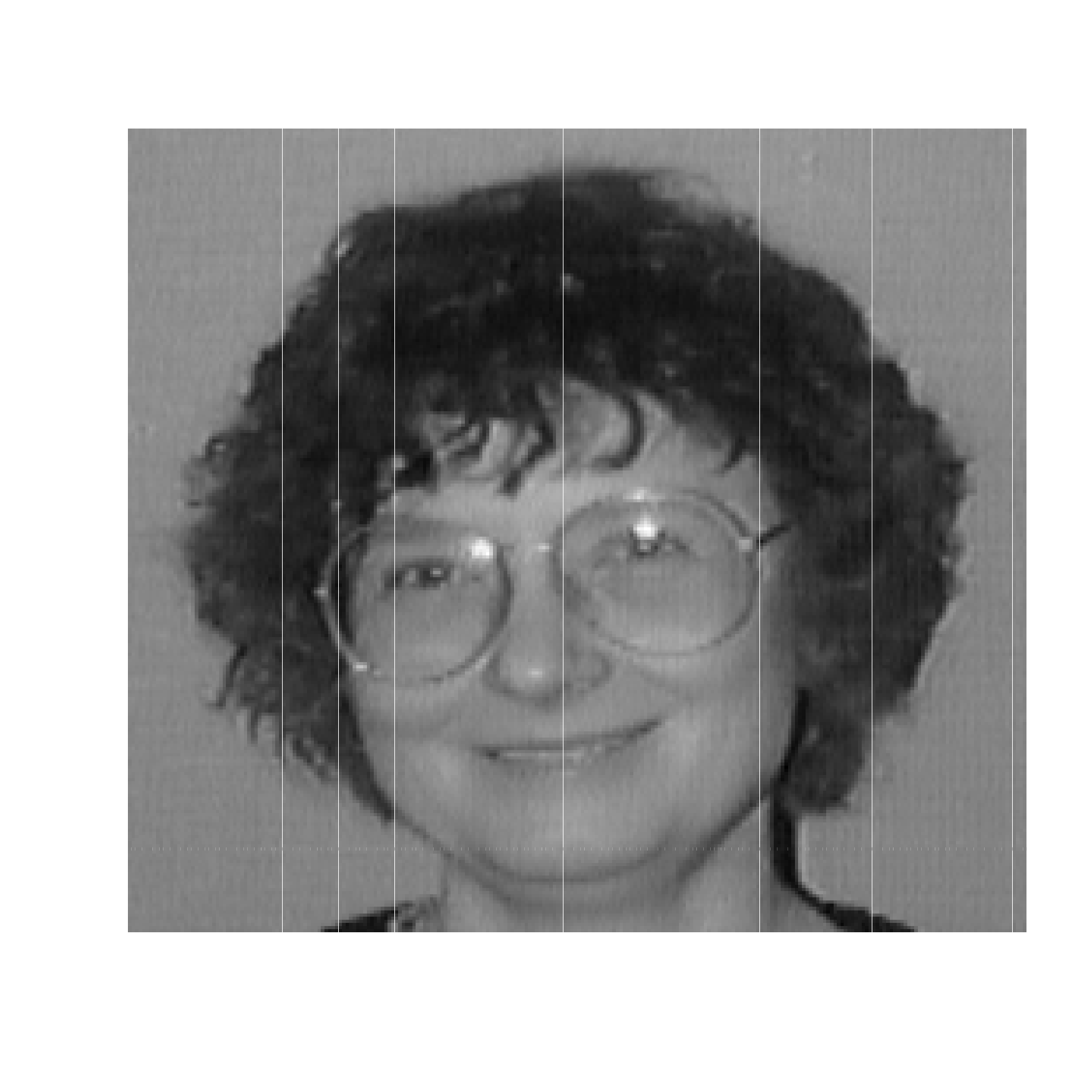}
\includegraphics[height=6cm,width=5cm]{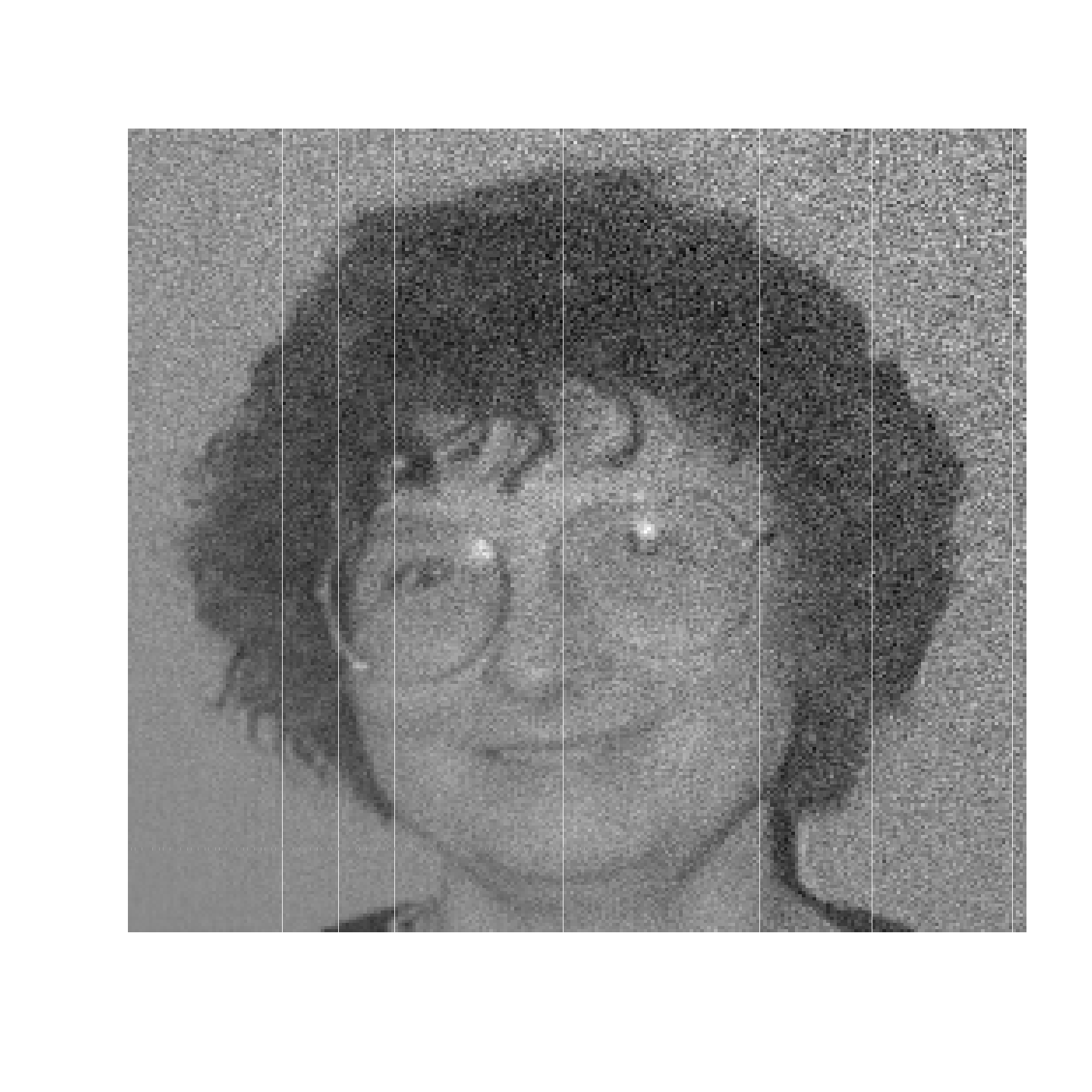}
\includegraphics[height=6cm,width=5cm]{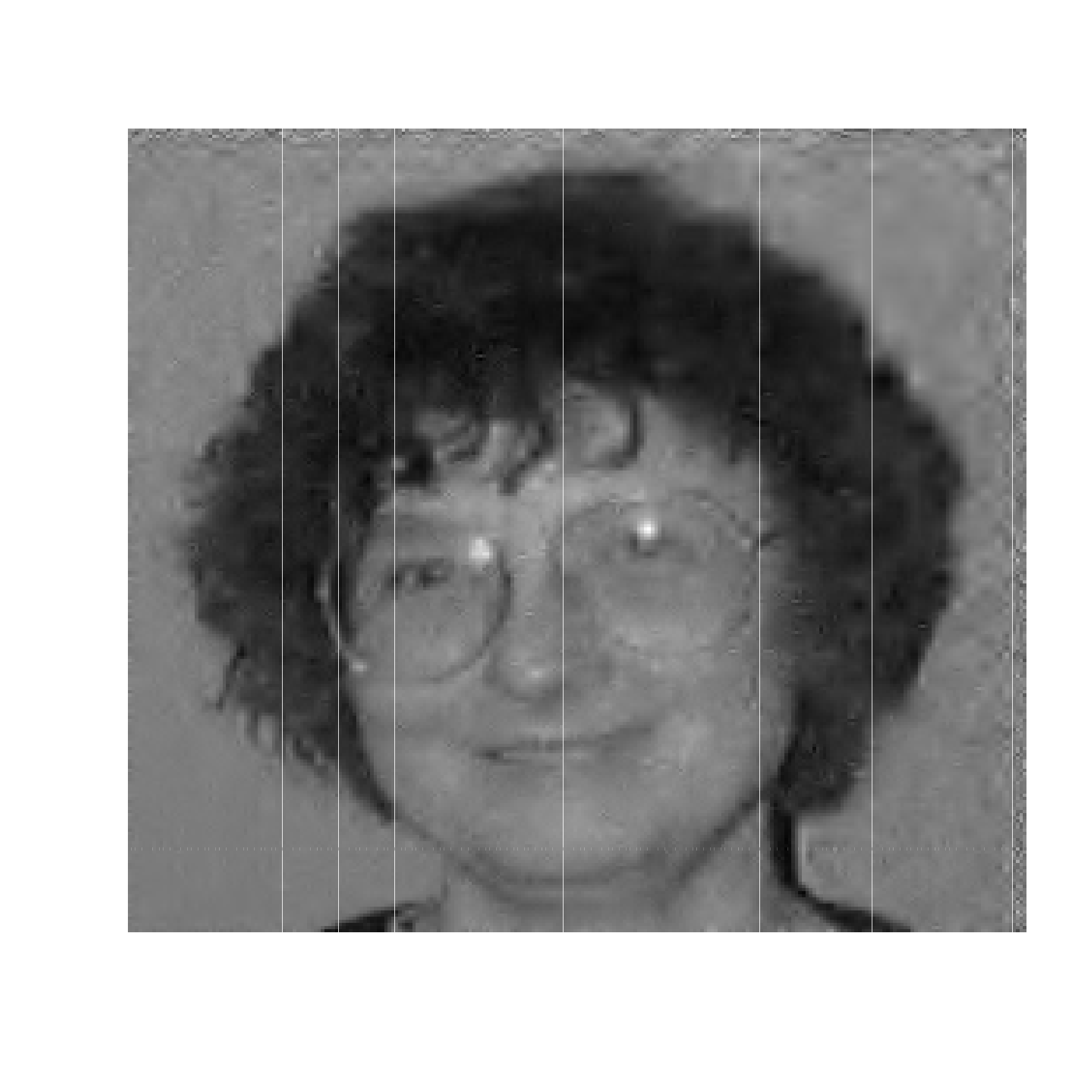}
\caption{Original image (left), noisy image (middle) and reconstructed image based on the posterior mean from the proposed method (right) of Ingrid
Daubechies, where $a_0=15$.}\label{fig:dau}
\end{figure}


\begin{table}[!h]
\caption{\small Ratio of the MSE
of the proposed method and NPMLE to that of \citet{johnstone2005empirical}.}\label{tb-dau}
\begin{center}
\begin{tabular}{rrrrr}
\hline
& &\multicolumn{3}{c}{$a_0$}\\ \cmidrule(r){3-5}
$m$ & method & 10 & 15 & 20 \\
\hline
$8$  & Semi  (mean)  & 0.820 & 0.840 & 0.850  \\
$8$  & Semi  (median)& 0.864 & 0.909 & 0.943  \\
$8$  & NPMLE (mean)  & 0.843 & 0.863 & 0.865  \\
$8$  & NPMLE (median)& 0.909 & 0.937 & 0.939  \\
$16$  & Semi  (mean) & 0.814 & 0.845 & 0.857  \\
$16$ & Semi  (median)& 0.856 & 0.911 & 0.951  \\
$16$  & NPMLE (mean) & 0.825 & 0.859 & 0.861  \\
$16$  & NPMLE(median)& 0.895 & 0.931 & 0.941  \\
\hline
\end{tabular}
\end{center}
\end{table}

\section{Appendix}
\subsection{Closed-form representations for posterior median}\label{closed-form}
\textbf{Double exponential}: We provide the closed-form representation
for $\delta(x;w,b,c)$ when the prior density component is double
exponential with location shift. We derive the result under the normal model $X|\mu\sim N(\mu,\sigma^2)$. Let
$g_+(x;b,c)=(1/\sigma)\int^{+\infty}_{c}\phi((x-\mu)/\sigma)\gamma(\mu;b,c)d\mu=(b/2)\exp(-bx+b^2\sigma^2/2+cb)\Phi((x-c)/\sigma-b\sigma)$
and $\tilde{g}_+(x;b,c)=(b/2)\exp(-bx+b^2\sigma^2/2+cb)\Phi(x/\sigma-b\sigma)$. Here we suppress the dependence on $\sigma$. Note that,
\begin{align*}
& \frac{1}{\sigma}\int^{+\infty}_{a}\phi\left(\frac{x-\mu}{\sigma}\right)\gamma(\mu;b,c)d\mu=(b/2)\exp(-bx+b^2\sigma^2/2+cb)\Phi\left(\frac{x-a}{\sigma}-b\sigma\right),
\end{align*}
for $a>c$, and
\begin{align*}
\frac{1}{\sigma}\int^{c}_{a}\phi\left(\frac{x-\mu}{\sigma}\right)\gamma(\mu;b,c)d\mu=(b/2)\exp(bx+b^2\sigma^2/2-cb)\left\{\Phi\left(\frac{c-x}{\sigma}-b\sigma\right)-\Phi\left(\frac{a-x}{\sigma}-b\sigma\right)\right\},
\end{align*}
for $a\leq c$. Then we
have $g(x;b,c)=g_+(x;b,c)+g_+(-x;b,-c),$ where $g$ denotes the convolution between $\phi_{\mu,\sigma^2}(\cdot)$ and $\gamma(\cdot;b,c)$.
Recall that $m(x;w,b,c)=(1-w)\phi_{0,\sigma^2}(x)+wg(x;b,c).$ Assuming $x>0$,
straightforward but tedious calculation shows that:
\\Case 1: if $c>0$ and $2wg_+(x;b,c)\geq m(x;w,b,c)$,
$$\delta(x;w,b,c)=x-b\sigma^2+\sigma\Phi^{-1}\left(1-\frac{m(x;w,b,c)\Phi((x-c)/\sigma-b\sigma)}{2wg_+(x;b,c)}\right).$$
Case 2: if $c>0$ and
$$1-\frac{2wg_+(-x,b,-c)\{\Phi((c-x)/\sigma-b\sigma)-\Phi(-x/\sigma-b\sigma)\}}{m(x;w,b,c)\Phi((c-x)/\sigma-b\sigma)}\leq \frac{2wg_+(x;b,c)}{m(x;w,b,c)}<1,$$
then
$$\delta(x;w,b,c)=x+b\sigma^2+\sigma\Phi^{-1}\left(\frac{\Phi((c-x)/\sigma-b\sigma)}{g_+(-x,b,-c)}\left\{g(x;b,c)-\frac{m(x;w,b,c)}{2w}\right\}\right).$$
Case 3: if $c>0$ and
$$\frac{2wg_+(x;b,c)}{m(x;w,b,c)}< 1-\frac{2wg_+(-x,b,-c)\{\Phi((c-x)/\sigma-b\sigma)-\Phi(-x/\sigma-b\sigma)\}}{m(x;w,b,c)\Phi((c-x)/\sigma-b\sigma)},$$
then $\delta(x;w,b,c)=0.$
\\Case 4: if $c\leq 0$ and $$\frac{2w\tilde{g}_+(x;b,c)}{m(x;w,b,c)}\geq 1,$$
then
$$\delta(x;w,b,c)=x-b\sigma^2+\sigma\Phi^{-1}\left(1-\frac{m(x;w,b,c)\Phi((x-c)/\sigma-b\sigma)}{2wg_+(x;b,c)}\right).$$
\\Case 5: if $c\leq 0$ and
$$1-\frac{2(1-w)\phi_{0,\sigma^2}(x)}{m(x;w,b,c)}\leq \frac{2w\tilde{g}_+(x;b,c)}{m(x;w,b,c)}< 1,$$
then $\delta(x;w,b,c)=0.$
\\Case 6: if $c\leq 0$ and
$$\frac{2w\tilde{g}_+(x;b,c)}{m(x;w,b,c)}<1-\frac{2(1-w)\phi_{0,\sigma^2}(x)}{m(x;w,b,c)},$$
then
\begin{align*}
\delta(x;w,b,c)=&x-b\sigma^2+\sigma\Phi^{-1}\bigg(\Phi(b\sigma-x/\sigma)
\\&-\left(\frac{m(x;w,b,c)}{2w}-\frac{w\tilde{g}_+(x;b,0)+(1-w)\phi_{0,\sigma^2}(x)}{w}\right)\frac{\Phi((x-c)/\sigma-b\sigma)}{g_+(x;b,c)}\bigg).
\end{align*}
Finally for $x<0$, we define
$\delta(x;w,b,c)=-\delta(-x;w,b,-c)$, i.e.,
$$\delta(x;w,b,c)=\text{sign}(x)\delta(|x|;w,b,\text{sign}(x)c),\quad x\in\mathbb{R}.$$

\begin{proof}[Proof of Lemma \ref{lemma:double}]
We prove the results when the noise level is $\sigma^2.$ Write
$\tau=1/(b^2\sigma^2)$. To show (1), first note that
$g(c,b,c)=2g_+(c,b,c)=b\exp(b^2\sigma^2/2)$ which is independent of
$c$, and $m(c;w,b,c)\rightarrow wg(c,b,c)$ as $c\rightarrow+\infty$.
Thus we have
$$\Phi^{-1}\left(\frac{\Phi(-b\sigma)}{g_+(c,b,c)}\left\{g(c;b,c)-\frac{m(c;w,b,c)}{2w}\right\}\right)\rightarrow -b\sigma.$$
By the closed-formed representation in Case 2, it is straightforward
to verify that $\delta(c;w,b,c)-c\rightarrow 0$ as $c\rightarrow
+\infty$.

Next we prove (2). As $x-c\rightarrow +\infty$ and
$x\rightarrow+\infty$, we have
$$\frac{\phi_{0,\sigma^2}(x)}{\exp(-bx+b^2\sigma^2/2+cb)}\rightarrow 0,\quad \frac{g_{+}(-x;b,-c)}{\exp(-bx+b^2\sigma^2/2+cb)}\rightarrow 0.$$
It thus implies that
\begin{align*}
&\frac{m(x;w,b,c)\Phi((x-c)/\sigma-b\sigma)}{2wg_+(x;b,c)}
=\frac{(1-w)\phi_{0,\sigma^2}(x)+wg_+(x;b,c)+wg_{+}(-x;b,-c)}{wb\exp(-bx+b^2\sigma^2/2+cb)}\rightarrow
1/2.
\end{align*}
When $c>0$, by Case 1, we have
$\delta(x;w,b,c)-(x-b\sigma^2)\rightarrow 0.$ When $c<0$, we have
$g_{+}(x;b,c)/\tilde{g}_+(x;b,c)\rightarrow 1.$ By Case 4, we have
$\delta(x;w,b,c)-(x-b\sigma^2)\rightarrow 0.$

Finally, (3) follows from similar argument and the fact that
$\delta(x;w,b,c)=-\delta(-x;w,b,-c)$ for $x<0.$
\end{proof}

\textbf{Normal}: Next we provide the closed-form representation
for $\delta(x;w,b,c)$ when the prior density component is normal
with location shift. The prior distribution for $\mu$ is
$(1-w)\delta_0+wN(c,1/b^2)$. Let $\tau=1/(b^2\sigma^2)$. Direct
calculation shows that
\begin{align*}
&h(x;a,b,c)=\int^{+\infty}_{a}\phi_{\mu,\sigma^2}(x)\gamma(\mu;b,c)d\mu=\phi_{c,1/b^2+\sigma^2}(x)\left\{1-\Phi\left(\frac{a-(\tau x+c)/(\tau+1)}{\sqrt{\sigma^2\tau/(1+\tau)}}\right)\right\},\\
&m(x;w,b,c)=(1-w)\phi_{0,\sigma^2}(x)+w\phi_{c,1/b^2+\sigma^2}(x).
\end{align*}
We have the following three cases:
\\Case 1: If $2wh(x;0,b,c)\geq m(x;w,b,c)$, then
$$\delta(x;w,b,c)=\frac{\tau x+c}{\tau+1}+\sigma\sqrt{\frac{\tau}{1+\tau}}\Phi^{-1}\left(\frac{w\phi_{c,1/b^2+\sigma^2}(x)-(1-w)\phi_{0,\sigma^2}(x)}{2w\phi_{c,1/b^2+\sigma^2}(x)}\right).$$
\\Case 2: If $m(x;w,b,c)-2(1-w)\phi_{0,\sigma^2}(x)\leq 2wh(x;0,b,c)\leq m(x;w,b,c)$, then $\delta(x;w,b,c)=0$.
\\Case 3: If $2wh(x;0,b,c)\leq m(x;w,b,c)-2(1-w)\phi_{0,\sigma^2}(x),$ then
$$\delta(x;w,b,c)=\frac{\tau x+c}{\tau+1}+\sigma\sqrt{\frac{\tau}{1+\tau}}\Phi^{-1}\left(\frac{w\phi_{c,1/b^2+\sigma^2}(x)+(1-w)\phi_{0,\sigma^2}(x)}{2w\phi_{c,1/b^2+\sigma^2}(x)}\right).$$

\begin{proof}[Proof of Lemma \ref{lemma:normal}]
Using the explicit expression for $\delta(x;w,b,c)$ and the fact
that $\phi_{0,\sigma^2}(x)/\phi_{c,1/b^2+\sigma^2}(x)\rightarrow 0$
as $|x|\rightarrow +\infty$, we have
\begin{align*}
\delta(x;w,b,c)-\frac{\tau x+c}{\tau+1}\rightarrow 0,
\end{align*}
as $|x|\rightarrow +\infty$.
\end{proof}

\subsection{Properties of the posterior median}
We present the following lemma which will be useful in the proof of
Proposition \ref{prop1}.
\begin{lemma}\label{lemma1}
For any $c\geq 0$, $g(x;c)/\phi(x)$ is monotonic increasing for $x>c$.
\end{lemma}

\begin{proof}[Proof of Lemma \ref{lemma1}]
Let $h(x,\mu;c)=\{\phi(x-\mu)+\phi(x+\mu-2c)\}/\phi(x)$. Then we
have
$g(x;c)/\phi(x)=\int^{+\infty}_{c}h(x,\mu;c)\gamma_0(\mu-c)d\mu$.
For $x>c$ and any $\mu$, we have
\begin{align*}
\frac{\partial h(x,\mu;c)}{\partial
x}=\mu\exp\{x\mu-\mu^2/2\}+(2c-\mu)\exp\{(2c-\mu)(2x+\mu-2c)/2\}.
\end{align*}
When $\mu>2c$, we have $\mu>(\mu-2c)\exp\{2(\mu-c)(c-x)\}$ which implies that
$\frac{\partial h(x,\mu;c)}{\partial
x}>0$. When $\mu\leq 2c$, it is clear that $\frac{\partial h(x,\mu;c)}{\partial
x}>0$. Therefore $g(x;c)/\phi(x)$ is monotonic increasing for $x>c\geq 0.$
\end{proof}



\begin{proof}[Proof of Proposition \ref{prop1}]
Without loss of generality, we assume that $c>0.$ Claim (1) follows
from the argument in the proof of Lemma 2 in JS (2004).

Under Condition (\ref{eq-prop1}), it is straightforward to verify that $\delta(0;w,c)=0$. By the monotonicity of $\delta$, there exist $t_1,t_2\geq 0$ such that
\begin{align*}
\delta(x;w,c)\begin{cases}
  <0, & \mbox{if }~~x<-t_2,\\
  =0, & \mbox{if }~~-t_2\leq x\leq t_1,\\
  >0, & \mbox{otherwise}.
\end{cases}
\end{align*}
Because
$P(\mu>0|X=x)=w\int_{0}^{+\infty}\phi(x-\mu)\gamma(\mu,c)d\mu/\{(1-w)\phi(x)+wg(x;c)\}$
and
$P(\mu<0|X=x)=w\int_{-\infty}^{0}\phi(-x+\mu)\gamma(\mu,c)d\mu/\{(1-w)\phi(x)+wg(x;c)\}$,
$t_1$ and $t_2$ must satisfy (\ref{median1}) and (\ref{median2}).

To show (3), first assume that $c>0$. We note that $\gamma(\mu;c)$ is symmetric about $c$ and
is unimodal. For $x>c>0$, we have $\gamma(x-v;c)\geq \gamma(x+v;c)$
for any $v\geq 0$. Thus we get
\begin{align*}
\gamma(x-v;c)\phi(v)/g(x;c)\geq
\gamma(x+v;c)\phi(v)/g(x;c).
\end{align*}
Integrating over $v\geq 0$, we obtain
\begin{align*}
P(\mu\leq x|X=x,\mu\neq 0)\geq P(\mu> x|X=x,\mu\neq 0).
\end{align*}
Because $P(\mu>x|X=x)=P(\mu>x|X=x,\mu\neq 0)P(\mu\neq 0|X=x)\leq
P(\mu>x|X=x,\mu\neq 0)\leq 0.5,$ we know that $\delta(x;w,c)\leq x.$
Similar argument shows that $\delta(x;w,c)\geq x$ for $x<0.$ Next we
consider the region where $x\leq c$. Using the fact that
$\phi(x-c-\mu)\leq \phi(x-c+\mu)$ for $x<c$ and any $\mu>0$, we deduce
that
\begin{align*}
P(\mu> c|X=x,\mu\neq
0)=&\int^{+\infty}_{c}\phi(x-\mu)\gamma_0(\mu-c)/g(x;c)d\mu=\int^{+\infty}_{0}\phi(x-c-\mu)\gamma_0(\mu)/g(x;c)d\mu
\\ \leq& \int^{+\infty}_{0}\phi(x-c+\mu)\gamma_0(\mu)/g(x;c)d\mu=\int_{-\infty}^{c}\phi(x-\mu)\gamma(\mu;c)/g(x;c)d\mu
\\=& P(\mu\leq c|X=x,\mu\neq 0),
\end{align*}
which implies that $P(\mu\geq c|X=x)\leq P(\mu> c|X=x,\mu\neq 0)\leq
0.5$ and thus $\delta(x;c)\leq c$. Therefore for $c>0$,
$|\delta(x;w,c)|\leq |x|\vee c$. Claim (3) follows by noticing that
$\delta(x;w,c)=-\delta(-x;w-c).$

Finally we prove (4). The proof is presented in four steps below.
\\\emph{Step 1}: Our arguments in Steps 1-3 are basically modifications of those in JS
(2004). We present the details for completeness. Assume that $c>0.$
Following the proof of Lemma 2 in JS (2004), we aim to find a
constant $a$ such that for large enough $x$,
\begin{align}\label{prop1-eq1}
P(\mu>x-a|X=x)=P(\mu>x-a|X=x,\mu\neq 0)P(\mu\neq 0|X=x)>1/2.
\end{align}
Let $B=\sup_{|u|\leq M}\gamma_0(u)e^{\Lambda
u}/\{\gamma_0(M)e^{\Lambda M}\}$. Under (\ref{eq-lem2}),
$\gamma_0(u)e^{\Lambda u}$ is increasing for $u\leq 0$ or $u\geq M$. Thus for any $a_1>M+c$,
we have
\begin{align*}
\text{Odd}(\mu>a_1|X=x,\mu\neq 0):=&\frac{P(\mu>a_1|X=x,\mu\neq
0)}{1-P((\mu>a_1|X=x,\mu\neq
0))}=\frac{\int_{a_1}^{+\infty}\gamma_0(\mu-c)\phi(x-\mu)d\mu}{\int_{-\infty}^{a_1}\gamma_0(\mu-c)\phi(x-\mu)d\mu}
\\ \geq & \frac{\int_{a_1}^{+\infty}e^{-\Lambda\mu}\phi(x-\mu)d\mu}{B\int_{-\infty}^{a_1}e^{-\Lambda\mu}\phi(x-\mu)d\mu}.
\end{align*}
Because $\int_{-\infty}^{+\infty}e^{-\Lambda \mu}\phi(\mu)d\mu<\infty,$
there exists a large enough $a_2>0$ such that
$\int_{-a_2}^{+\infty}e^{-\Lambda\mu}\phi(\mu)d\mu>3B\int_{-\infty}^{-a_2}e^{-\Lambda\mu}\phi(\mu)d\mu.$
Thus for $x>a_1+a_2+M$, we have
\begin{align*}
\text{Odd}(\mu>x-a_1|X=x,\mu\neq 0)\geq
\frac{\int_{x-a_1}^{+\infty}e^{-\Lambda\mu}\phi(x-\mu)d\mu}{B\int_{-\infty}^{x-a_1}e^{-\Lambda\mu}\phi(x-\mu)d\mu}
=\frac{\int_{-a_1}^{+\infty}e^{-\Lambda\mu}\phi(\mu)d\mu}{B\int_{-\infty}^{-a_1}e^{-\Lambda\mu}\phi(\mu)d\mu}>3,
\end{align*}
which implies that $P(\mu>x-a_1|X=x,\mu\neq 0)>3/4.$
\\\emph{Step 2}: The posterior odds $\text{Odd}(\mu\neq 0|X=x)$ is equal to
\begin{align*}
O(x;w,c):=\text{Odd}(\mu\neq 0|X=x)=\frac{P(\mu\neq
0|X=x)}{1-P(\mu\neq 0|X=x)}=\frac{w}{1-w}\frac{g(x;c)}{\phi(x)}.
\end{align*}
Let $w_c=\{\phi(c)/g(c;c)\}/[1+\{\phi(c)/g(c;c)\}]$ so that
$O(c;w_c,c)=1.$ For fixed $w_c$, by Lemma \ref{lemma1}, $O(x;w_c,c)$
is an increasing function from $1$ to $+\infty$ when $x\geq c.$ For
$w<w_c$, there exists a $e(w)>c$ such that $O(e(w);w,c)=1.$ Also
note that
\begin{align*}
O(x;w,c)=O(x_0;w,c)\exp\left\{\int^{x}_{x_0}(\log(g(\mu;c))'-\log(\phi(\mu))')d\mu\right\}.
\end{align*}
\\\emph{Step 3}: Let $\epsilon=(\rho-\Lambda)/2$. Note that $g(x;c)=\int^{+\infty}_{-\infty}\gamma(\mu-c)\phi(x-\mu)d\mu
=\int^{+\infty}_{-\infty}\gamma(\mu)\phi(x-c-\mu)d\mu.$ For large
enough $a_3>M+c$, we have for $|\mu|\geq a_3$,
\begin{align*}
(\log g(\mu;c))'\geq -\Lambda-\epsilon,\quad (\log \phi(\mu))'\leq
-\rho,
\end{align*}
where we have used (31) in JS (2004). Choose $w_3$ so that
$O(a_3;w_3,c)=1.$ For $w<w_3$, $e(w)>a_3.$ For $x>e(w)+a_4$ with
$a_4=2(\rho-\Lambda)^{-1}\log(2)$, we have
\begin{align*}
O(x;w,c)\geq O(e(w);w,c)\exp\{(\rho-\Lambda)a_4/2\}\geq 2.
\end{align*}
If $w\geq w_c$, then $O(x;w,c)\geq O(x;w_c,c)\geq 2$ provided that
$x>c+a_4.$ In either cases, it follows that $P(\mu\neq 0|X=x)\geq
2/3.$ Therefore when $x>\max\{c+a_4,e(w)+a_4,a_1+a_2+M\},$
(\ref{prop1-eq1}) holds with $a=a_1$. If $0\leq
x<\max\{c+a_4,e(w)+a_4,a_1+a_2+M\},$ we have $0\leq
\delta(x;w,c)<x\vee c$ by Claim (3). We also note that $e(w)\leq
t_1$. Simple algebra shows that $O(e(w),w,c)=1$ implies
$$\frac{w\int_{0}^{+\infty}\gamma_0(e(w)-\mu)\phi(\mu)d\mu}{(1-w)\phi(e(w))+g(e(w);c)}\leq 0.5.$$
Thus we have $\delta(e(w);w,c)\leq \delta(t_1;w,c)=0$ which suggests that $e(w)\leq t_1$ as $\delta(\cdot;w,c)$ is a monotonic increasing function.
Combining the arguments we get
$$-c\leq x-\delta(x;w,c)\leq t_1+c+c_0,$$
for some constant $c_0$.\\
\emph{Step 4}: For $c>0$ and $x<0,$ we want to find a positive
constant $a$ such that
\begin{align}\label{step5-1}
P(\mu>x+a|X=x)<1/2.
\end{align}
It thus implies that $0\leq \delta(x;w,c)-x\leq a$. First note that for $x<-a$, (\ref{step5-1}) is equivalent to
\begin{align}\label{step5-2}
(1-w)\phi(x)+w\int^{+\infty}_{x+a}\phi(x-\mu)\gamma_0(\mu-c)d\mu\leq \frac{1}{2}m(x;w,c).
\end{align}
Rearranging (\ref{step5-2}), we have
\begin{align}\label{step5-3}
\frac{(1-w)\phi(x)}{wg(x;c)}+\frac{2\int^{+\infty}_{x+a}\phi(x-\mu)\gamma_0(\mu-c)d\mu}{g(x;c)}\leq 1.
\end{align}
Using the fact that $g(x;c)\geq c_0 \gamma(x-c)$ [see (28) of JS (2004)], for any $\epsilon>0,$ there exists $x<-c$ such that,
\begin{align*}
\frac{(1-w)\phi(x)}{wg(x;c)}\leq \frac{(1-w)\phi(x)}{c_0w\gamma(x-c)}\leq \epsilon.
\end{align*}
The second term on the LHS in (\ref{step5-3}) is a monotonic increasing function of
\begin{align}\label{step5-4}
\frac{\int^{+\infty}_{x+a}\phi(x-\mu)\gamma_0(\mu-c)d\mu}{\int^{x+a}_{-\infty}\phi(x-\mu)\gamma_0(\mu-c)d\mu}.
\end{align}
When $x< -a-M$, (\ref{step5-4}) can be bounded by
\begin{align*}
&\frac{\int^{+\infty}_{-M+c}\phi(x-\mu)\gamma_0(\mu-c)d\mu+\int^{-M+c}_{x+a}\phi(x-\mu)\gamma_0(\mu-c)d\mu}{\int^{x+a}_{-\infty}\phi(x-\mu)\gamma_0(\mu-c)d\mu}
\\ \leq &\frac{\phi(x+M-c)+\int^{-M+c}_{x+a}\phi(x-\mu)\gamma_0(\mu-c)d\mu}{\int^{x+a}_{-\infty}\phi(x-\mu)\gamma_0(\mu-c)d\mu}
\\ = &\frac{\phi(x+M-c)+\int^{-M}_{x_0+a}\phi(x_0-\mu)\gamma_0(\mu)d\mu}{\int^{x_0+a}_{-\infty}\phi(x_0-\mu)\gamma_0(\mu)d\mu}
\\ = &\frac{\phi(x+M-c)+\int^{y_0-a}_{M}\phi(\mu-y_0)\gamma_0(\mu)d\mu}{\int^{+\infty}_{y_0-a}\phi(\mu-y_0)\gamma_0(\mu)d\mu},
\end{align*}
where $x_0=x-c$ and $y_0=-x_0=c-x$. For $M<\mu\leq y_0-a$, $\gamma_0(\mu)e^{\Lambda\mu}\leq \gamma_0(y_0-a)e^{\Lambda(y_0-a)}$.
For $\mu>y_0-a$, $\gamma_0(\mu)e^{\Lambda\mu}\geq \gamma_0(y_0-a)e^{\Lambda(y_0-a)}$. Thus we have
\begin{align*}
&\frac{\phi(x+M-c)+\int^{y_0-a}_{M}\phi(\mu-y_0)\gamma_0(\mu)d\mu}{\int^{+\infty}_{y_0-a}\phi(\mu-y_0)\gamma_0(\mu)d\mu}
\\ \leq &\frac{\phi(x+M-c)e^{-\Lambda(y_0-a)}/\gamma_0(y_0-a)+\int^{y_0-a}_{M}\phi(\mu-y_0)e^{-\Lambda\mu}d\mu}{\int^{+\infty}_{y_0-a}\phi(\mu-y_0)e^{-\Lambda\mu}d\mu}
\\ \leq &\frac{\phi(x+M-c)e^{-\Lambda(y_0-a)}/\gamma_0(y_0-a)+(\Phi(\Lambda-a)-\Phi(\Lambda+M-y_0))e^{-y_0\Lambda+\Lambda^2/2}}{(1-\Phi(\Lambda-a))e^{-y_0\Lambda+\Lambda^2/2}}
\\ \leq &\frac{\phi(x+M-c)e^{\Lambda a-\Lambda^2/2}/\gamma_0(y_0-a)+\{\Phi(\Lambda-a)-\Phi(\Lambda+M-y_0)\}}{\{1-\Phi(\Lambda-a)\}}.
\end{align*}
Note that as $x\rightarrow -\infty,$ $\phi(x-M)e^{y_0\Lambda-\Lambda^2/2}\rightarrow0$ and $\Phi(\Lambda+M-y_0)\rightarrow 0.$ Also we can make $\Phi(\Lambda-a)$ small by picking a large enough $a$.
Combining the above derivations, there exists a $c_2>0$ such that for $x<-c-c_2$, (\ref{step5-1}) holds and thus $0\leq \delta(x;w,c)-x\leq a$. When $-c-c_2\leq x<-t_2$, $0\leq \delta(x;w,c)-x\leq c+c_2$. When $-t_2\leq x\leq 0,$
$\delta(x;w,c)-x=-x\leq t_2.$ The proof is completed by noticing $\delta(x;w,c)=\text{sign}(x)\delta(|x|;w,\text{sign}(x)c)$.
\end{proof}

\begin{proof}[Proof of Lemma \ref{lemma:con}]
By the definition of $\delta^{-1}$, we have $\lim_{t\rightarrow
0^+}\delta^{-1}(t;w,b,c)=t_1$ and $\lim_{t \rightarrow
0^-}\delta^{-1}(t;w,b,c)=-t_2,$ which implies that
$\lim_{\theta\rightarrow 0^+}\mathcal{P}'(\theta;w,b,c)=t_1$ and
$\lim_{\theta\rightarrow 0^-}\mathcal{P}'(\theta;w,b,c)=-t_2$ with
$\mathcal{P}'=\partial \mathcal{P}/\partial \theta.$ We first argue
that the solution to (\ref{pen}) is a thresholding rule. Note the
first derivative of (\ref{pen}) with respect to $\theta$ is
$l'(\theta,x):=\text{sign}(\theta)\{|\theta|+\text{sign}(\theta)\mathcal{P}'(\theta;w,b,c)\}-x$.
Therefore for $-t_2<x<t_1$, $l'(\theta,x)>0$ for small enough
positive $\theta$, and $l'(\theta,x)<0$ for large enough negative
$\theta$. Hence, $\hat{\theta}(x;w,b,c)=0$ for $-t_2<x<t_1$. For $x>t_1$ or $x<-t_2$,
the unique solution to the equation $l'(\theta,x)=0$ satisfies
\begin{align*}
\theta+\mathcal{P}'(\theta;w,b,c)=\theta+\{\delta^{-1}(\theta;w,b,c)-\theta\}=x,
\end{align*}
which implies that $\hat{\theta}=\delta(x;w,b,c)$.
\end{proof}

\subsection{Proof of Theorem \ref{thm-main}}
By Stein's Lemma, SURE can also be written as
\begin{align*}
\hat{R}(w,c)=\hat{R}(\theta)=\sum_{i=1}^{p}(\zeta(X_i;\theta)-X_i)^2+2\sum_{i=1}^{p}\nabla\zeta(X_i;\theta)-p,
\end{align*}
which is more convenient for our theoretical analysis. Consider
\begin{align*}
\frac{1}{p}\left\{\hat{R}(w,c)-\E\hat{R}(w,c)\right\}=&\frac{1}{p}\sum_{i=1}^{p}\left\{(\zeta(X_i;\theta) - X_i)^2-\E(\zeta(X_i;\theta) - X_i)^2\right\}
\\&+ \frac{2}{p}\sum_{i=1}^p \left(\nabla \zeta(X_i;\theta) - \E \nabla \zeta(X_i;\theta)\right)=\frac{1}{p}\sum_{i=1}^{p}W_i,
\end{align*}
where $W_i=(\zeta(X_i;\theta) - X_i)^2-\E(\zeta(X_i;\theta) -
X_i)^2+2\{\nabla \zeta(X_i;\theta) - \E \nabla \zeta(X_i;\theta)\}.$

We first state the following lemma, which shows the bounded
shrinkage property for the posterior mean. Recall that
$\gamma_0(u)=\gamma(u,1,0)$. Write $a\lesssim b$ if $a\leq Cb$ for
some constant $C$ which is independent of $p.$
\begin{lemma}\label{lemma:mean}
Assume that $\gamma_0$ is unimodal with
\begin{equation}\label{1}
\sup_{u}\left|\nabla \log\gamma_0(u)\right|\leq \Lambda \quad
{a.e.},
\end{equation}
for $\Lambda>0.$ Then we have for any $x\in\mathbb{R}$,
$$|\zeta(x;\theta)-x|\lesssim 1+\sqrt{|c|+\log(1/w)}.$$
\end{lemma}
\begin{proof}[Proof of Lemma \ref{lemma:mean}]
Note that $\partial \phi(x-u)/\partial x=-\partial
\phi(x-u)/\partial u$. Then we have
\begin{align*}
\nabla m(x;\theta)=&-(1-w)x\phi(x)-w\int \gamma(u;c)(\partial \phi(x-u)/\partial u) d u
\\=&-(1-w)x\phi(x)+w\int\phi(x-u) \nabla \gamma(u;c) d u
\\=&-(1-w)x\phi(x)+w\int \phi(x-u) \gamma(u;c)\nabla \log\gamma(u;c) d
u.
\end{align*}
As $|\nabla\log\gamma(u;c)| \leq \Lambda$, it is not hard to see
that
\begin{equation}\label{eq-bound0}
\left|\frac{w\int \phi(x-u) \gamma(u;c)\nabla \log\gamma(u;c) d
u}{m(x;\theta)}\right| \leq \frac{\Lambda wg(x;c)}{m(x;\theta)} \leq
\Lambda.
\end{equation}
In view of the proof of Lemma 1 in JS (2004), there exists $C_1>0$
such that for any $x,u>0$,
$$\gamma_0(x+u)\geq C_1e^{-\Lambda u}\gamma_0(x).$$
Let $x^*=x-c$. We have for $x^*>0$,
\begin{align*}
g(x;c)=& \int \phi(x^*-u)\gamma_0(u)du
 \geq  \int_{0}^{\infty}\phi(u)\gamma_0(x^*+u)du
 \geq C_1\int_{0}^{\infty}\phi(u)\gamma_0(x^*)e^{-\Lambda u}du,
\end{align*}
and for $x^*<0$,
\begin{align*}
g(x;c)=& \int \phi(x^*-u)\gamma_0(u)du
\geq \int_{0}^{\infty} \phi(u)\gamma_0(u-x^*)du
\geq C_1\int_{0}^{\infty}\phi(u)\gamma_0(x^*)e^{-\Lambda u}du.
\end{align*}
Under (\ref{1}), there exists a constant $C_2$ such that $C_2
e^{-\Lambda|x|}\leq \gamma_0(x)$ for any $x$. Together with
(\ref{eq-bound0}), we have
\begin{align}
|\zeta(x;\theta)-x|\leq & \left|\frac{(1-w)x\phi(x)}{m(x;\theta)}\right|+\Lambda \leq \frac{(1-w)|x|}{(1-w)+wC_3e^{x^2/2-\Lambda|x-c|}}+\Lambda \nonumber
\\ \leq & \frac{(1-w)|x|}{(1-w)+C_3e^{x^2/2-\Lambda|x|-\Lambda|c|-\log(1/w)}}+\Lambda \nonumber
\\ \leq & \frac{(1-w)(|x|-\Lambda|+\Lambda)}{(1-w)+C_3e^{(|x|-\Lambda)^2/2-\Lambda|c|-\log(1/w)}}+\Lambda, \label{bound}
\end{align}
where $C_3>0$ is a constant which could be different from line to line. When $(|X|-\Lambda)^2\leq 4\Lambda |c|+4\log(1/w)$, the first term in (\ref{bound}) is bounded by $\Lambda+2\sqrt{\Lambda|c|+\log(1/w)}$.
When $(|X|-\Lambda)^2>4\Lambda |c|+4\log(1/w)$, the first term in (\ref{bound}) is bounded by $(|x|-\Lambda|+\Lambda)/\{C_3e^{(|x|-\Lambda)^2/4}\}\leq C_4$ for some $C_4>0$. Therefore, we have
$|\zeta(x;\theta)-x|\lesssim 1+\sqrt{|c|+\log(1/w)}$.
\end{proof}

\begin{lemma}\label{lemma:mean2}
Suppose the assumptions in Lemma \ref{lemma:mean} hold. Further assume that
\begin{align}\label{c1}
\sup_{u}|\nabla^2 \log\gamma_0(u)|\leq \Lambda',\quad \text{a.e.},
\end{align}
for some $\Lambda'>0$. Then we have for any $x\in\mathbb{R}$,
\begin{align}\label{11}
|\nabla \zeta(x;\theta)|\lesssim & 1+|c|+\log(1/w).
\end{align}
The same conclusion holds when $\gamma$ is double exponential.
\end{lemma}
\begin{proof}
Notice that
\begin{align*}
|\nabla \zeta(x;\theta)| \leq & 1+\left|\frac{\nabla^2m(x;\theta)}{m(x;\theta)}\right|+(\nabla \log m(x;\theta))^2. 
\end{align*}
Consider
\begin{align}
\nabla^2 m(x;\theta)=&-(1-w)\{\phi(x)-x^2\phi(x)\}-w\int \{\partial \phi(x-u)/\partial u\} \gamma(u;c)\nabla \log\gamma(u;c) du \nonumber
\\=&-(1-w)\{\phi(x)-x^2\phi(x)\}+w\int \phi(x-u)\gamma(u;c)(\nabla \log\gamma(u;c) )^2 du \nonumber
\\&+w\int \phi(x-u)\gamma(u;c)\nabla^2\log\gamma(u;c) du. \label{last-term}
\end{align}
Under the assumption that $\sup_{u}|\nabla^2\log\gamma_0(u)|\leq \Lambda'$, we see that
\begin{align}\label{bound-3}
\left|\frac{w\int
\phi(x-u)\gamma(u;c)\nabla^2\log\gamma(u;c)du}{m(x;\theta)}\right|\leq
\frac{\Lambda' wg(x;c)}{m(x;\theta)}\leq \Lambda'.
\end{align}
The rest of the proof is similar to those for Lemma
\ref{lemma:mean}. we skip the details here to conserve space.

The argument in (\ref{bound-3}) is not applicable to double
exponential distribution but the conclusion remains true. When
$\gamma_0(u)=\Lambda\exp(-\Lambda|u|)/2$, we have
$\nabla\log\gamma_0(u)=-\Lambda\text{sign}(u)$ and $\nabla^2 \log
\gamma_0(u)=-2\Lambda\delta(u)$, where $\delta(u)$ is the Dirac
Delta function. Then (\ref{last-term}) becomes
$$-2\Lambda w\int \phi(x-u)\gamma(u;c)\delta(u-c)du=-2\Lambda w \phi(x-c)\gamma_0(0),$$
which is bounded uniformly over $x,c$ and $w$, when divided by
$m(x;\theta).$
\end{proof}

By Lemmas \ref{lemma:mean}-\ref{lemma:mean2}, we have
\begin{align*}
|W_i|\leq c_1 + c_2\{|c|+\log(1/w)\},
\end{align*}
for some positive constants $c_1,c_2>0.$ Applying the Hoeffding's
inequality to $p^{-1}\sum_{i=1}^{p}W_i$, we have for any
$\epsilon>0,$
\begin{align}\label{eq-hoff}
P\left(\frac{1}{p}|\hat{R}(w,c)-E\hat{R}(w,c)|>\frac{\epsilon}{\sqrt{p}}\right)\leq
2\exp\left[-\frac{2\epsilon^2}{\{c_1 +
c_2(|c|+\log(1/w))\}^2}\right].
\end{align}
For distinct $\theta'=(w',c')$ and $\theta=(w,c)$, we aim to bound
$|\hat{R}(w,c)-\hat{R}(w',c')|$. We assume that $w,w'\in
[1/\lambda_0,1]$ and $c,c'\in [-c_0,c_0]$ for $\lambda_0,c_0>0$,
where $\lambda_0$ and $c_0$ are allowed to grow with $p.$
The following equations are useful in the subsequent calculations,
\begin{equation}\label{eq-bound1}
\begin{split}
\nabla m(x;\theta)-\nabla m(x;\theta')
=&(w-w')x\phi(x)+ (w'-w)\int \gamma(u;c')(\partial \phi(x-u)/\partial u) d u
\\&+w\int (\gamma(u;c')-\gamma(u;c))(\partial \phi(x-u)/\partial u) d u,
\end{split}
\end{equation}
and
\begin{equation}\label{eq-bound2}
\begin{split}
&\nabla^2 m(x;\theta)-\nabla^2 m(x;\theta')
\\=& (w-w')\{\phi(x)-x^2\phi(x)\}+(w'-w)\int \{\partial \phi(x-u)/\partial u\} \gamma(u;c')\nabla \log\gamma(u;c') du
\\&+w\int \{\partial \phi(x-u)/\partial u\} \{\gamma(u;c')\nabla \log\gamma(u;c')-\gamma(u;c)\nabla \log\gamma(u;c)\} du.
\end{split}
\end{equation}

\subsubsection{Case 1: $c=0$}
To gain some insight, we focus on a simpler case where $c=0$. Note
that
\begin{align*}
&|\nabla \log m(x;\theta)-\nabla \log m(x;\theta')|
\\ =& \left|\frac{\nabla m(x;\theta)-\nabla m(x;\theta')}{m(x;\theta')}+\frac{\nabla \log m(x;\theta)}{m(x;\theta')}(m(x;\theta')-m(x;\theta))\right|
\\ \lesssim& |w-w'|/(ww')\leq \lambda_0^2|w-w'|,
\end{align*}
where we have used the fact that $|\nabla \log m(x;\theta)|\lesssim
1/w$, $|m(x;\theta)-m(x;\theta')|/m(x;\theta')\lesssim |w-w'|/w$,
and $|\nabla m(x;\theta)-\nabla m(x;\theta')|/m(x;\theta')\lesssim
|w-w'|/w'.$ Similarly, we can deduce that
\begin{align*}
&|\nabla^2 \log m(x;\theta)-\nabla^2 \log m(x;\theta')|
\\ =& \left|\frac{\nabla^2 m(x;\theta)}{m(x;\theta)}-\frac{\nabla^2 m(x;\theta')}{m(x;\theta')}\right|+|(\nabla \log m(x;\theta))^2-(\nabla \log m(x;\theta'))^2|
\\ \lesssim & \left|\frac{\nabla^2 m(x;\theta)}{m(x;\theta)m(x;\theta')}(m(x;\theta')-m(x;\theta))+\frac{\nabla^2 m(x;\theta)-\nabla^2 m(x;\theta')}{m(x;\theta')}\right|+\lambda_0^3|w-w'|
\\  \lesssim & \lambda_0^3|w-w'|.
\end{align*}
Thus we have
\begin{align*}
p^{-1}|\{\hat{R}(w,0)-E\hat{R}(w,0)\}-\{\hat{R}(w',0)-E\hat{R}(w',0)\}|\lesssim
\lambda_0^3|w-w'|.
\end{align*}
Now set $w_j=\delta j$ for $j=1,2,\dots$ such that $w_j \in
[1/\lambda_0,1]$. Choose $\delta$ so that $\delta
\lambda_0^3=o(1/\sqrt{p})$. Then we have
\begin{align*}
A=\left\{\max_{w\in [1/\lambda_0,1]}p^{-1}|\hat{R}(w,0)-E\hat{R}(w,0)|\geq 2\epsilon/\sqrt{p}\right\}\subseteq D,
\end{align*}
where
$$D=\left\{\max_{j}p^{-1}|\hat{R}(w_j,0)-E\hat{R}(w_j,0)|\geq \epsilon/\sqrt{p}\right\}.$$
Using the union bound and the Hoeffding's inequality in
(\ref{eq-hoff}) with $c=0$, we have for large enough $p$,
\begin{align}\label{1}
P\left(A\right)\leq P(D)\leq \frac{4(\lambda_0-1)}{\lambda_0\delta}\exp\left[-\frac{2\epsilon^2}{\{c_1 + c_2\log(\lambda_0)\}^2}\right].
\end{align}
Choosing $\epsilon^2=s^2\log (p)(c_1 + c_2 \log(\lambda_0))^2/2,$ we
obtain
$$P(A)\leq \frac{4(\lambda_0-1)}{\lambda_0\delta}p^{-s^2}.$$ This says that
\begin{align*}
P\left(\max_{w\in [1/\lambda_0,1]}\frac{1}{\sqrt{p\log (p)}(c_1+c_2\log(\lambda_0))}|\hat{R}(w,0)-E\hat{R}(w,0)|\geq \sqrt{2}s\right)\leq \frac{4(\lambda_0-1)}{\lambda_0\delta}p^{-s^2}.
\end{align*}
For example, with $\lambda_0=a_1p^{a_2}$, one can pick $\delta=1/p^{3a_2+1/2+\varepsilon}$ and large enough $s$, where $\varepsilon>0$. Then $\delta \lambda_0^3=o(1/\sqrt{p})$
and
$$\max_{w\in [1/\lambda_0,1]}p^{-1}|\hat{R}(w,0)-E\hat{R}(w,0)|=O_p\left(\frac{(\log(p))^{3/2}}{\sqrt{p}}\right).$$

\subsubsection{Case 2: general $c$}
Now we consider the general case: $c\in [-c_0,c_0]$, where $c_0$ is allowed to grow slowly with $p.$ In view of the proof of Case 1, we need to bound the following quantities:
\begin{align}
&\frac{m(x;\theta')-m(x;\theta)}{m(x;\theta')},\\
&\frac{\nabla m(x;\theta)-\nabla m(x;\theta')}{m(x;\theta')}, \\
&\frac{\nabla^2 m(x;\theta)-\nabla^2 m(x;\theta')}{m(x;\theta')}.
\end{align}
For clarity, we present the proof in the following 5 steps.

\noindent \emph{ Step 1:} We deal with the first quantity. By the triangle inequality,
\begin{align*}
|m(x;\theta)-m(x;\theta')|\leq |w-w'|\{\phi(x)+g(x;c')\}+w|g(x;c)-g(x;c')|.
\end{align*}
Notice that
\begin{align*}
|\log\gamma_0(u-c)-\log\gamma_0(u-c')|=\left|\int_{u-c}^{u-c'}\nabla \log \gamma_0(s)ds\right|\leq \Lambda |c-c'|.
\end{align*}
and $|e^x-1|\leq |x|e^{|x|}$ for any $x$. Using these facts, we get
\begin{align*}
|g(x;c)-g(x;c')|\leq& \int \phi(x-u)\gamma_0(u-c')|\gamma_0(u-c)/\gamma_0(u-c')-1|du
\\ =&  \int \phi(x-u)\gamma_0(u-c')|e^{\log\gamma_0(u-c)-\log\gamma_0(u-c')}-1|du
\\ \leq & g(u;c') \Lambda|c-c'| e^{\Lambda|c-c'|}.
\end{align*}
Combining these results, we have
\begin{align}
\left|\frac{m(x;\theta')-m(x;\theta)}{m(x;\theta')}\right|\lesssim |w-w'|e^{\Lambda|c'|}/w' + |c-c'|e^{\Lambda|c-c'|}/w',
\end{align}
where we use the bound $\phi(x)/m(x;\theta')\lesssim
e^{\Lambda|c'|}/w'$ uniformly over $x$.\footnote{This bound can be
improved if we are willing to assume an upper bound on $w$, i.e.,
$w\leq \tilde{c}<1$. In this case, $c_0$ is allowed to grow at a
faster rate.}

\noindent \emph{ Step 2:} To deal with the second quantity, we note that
\begin{align*}
&\left|\int (\gamma(u;c')-\gamma(u;c))(\partial \phi(x-u)/\partial u) d u\right|
\\ \leq &\int |\nabla \gamma(u;c)-\nabla \gamma(u;c')|\phi(x-u)du
\\=& \int \left|\nabla \log\gamma(u;c)\left(\frac{\gamma(u;c)}{\gamma(u;c')}-1\right)+\nabla \log\gamma(u;c)-\nabla \log \gamma(u;c')\right|\gamma(u;c')\phi(x-u)du.
\end{align*}
Then by (\ref{eq-bound1}) and similar argument as above, we obtain,
\begin{align*}
\left|\frac{\nabla m(x;\theta)-\nabla
m(x;\theta')}{m(x;\theta')}\right|&\lesssim
|w-w'|e^{\Lambda|c'|}/w'+|c-c'|e^{\Lambda |c-c'|}/w',
\end{align*}
where we have used the fact that $|\nabla \log\gamma(u;c)-\nabla \log \gamma(u;c')|=|\int^{u-c}_{u-c'}\nabla^2 \log \gamma_0(s)ds | \lesssim |c-c'|$.

\begin{remark}
{\rm For double exponential distribution distribution, we have
\begin{align*}
&\frac{\int \left|\nabla \log\gamma(u;c)-\nabla \log \gamma(u;c')\right|\gamma(u;c')\phi(x-u)du}{m(x;\theta')}
\\=&\frac{2\Lambda\left|\int_{c}^{c'}\gamma(u;c')\phi(x-u)du\right|}{m(x;\theta')}\lesssim \frac{ |c-c'|\phi(x-c^*)}{w'g(x;c')}
\\ \lesssim &\frac{ |c-c'|e^{-(x-c^*)^2/2+\Lambda|x-c'|}}{w'}
\\ \leq &\frac{ |c-c'|e^{-(x-c')^2/4+(c-c')^2/2+\Lambda|x-c'|}}{w'}
\\ \lesssim&  |c-c'|e^{(c-c')^2/2}/w',
\end{align*}
where $c^*$ is between $c$ and $c'.$ So we have
\begin{align}
\left|\frac{\nabla m(x;\theta')-\nabla
m(x;\theta)}{m(x;\theta')}\right|\lesssim |w-w'|e^{\Lambda|c'|}/w' +
|c-c'|e^{\Lambda|c-c'|+(c-c')^2/2}/w'.
\end{align}
}
\end{remark}

\noindent \emph{ Step 3:} Next we analyze the third quantity. In view of (\ref{eq-bound2}), we consider
\begin{align*}
&\int \{\partial \phi(x-u)/\partial u\} \{\gamma(u;c')\nabla \log\gamma(u;c')-\gamma(u;c)\nabla \log\gamma(u;c)\} du
\\=&\int \{\partial \phi(x-u)/\partial u\}\gamma(u;c')\{\nabla \log\gamma(u;c')-\nabla \log\gamma(u;c)\} du
\\&+\int \{\partial \phi(x-u)/\partial u\}\nabla \log\gamma(u;c)\{\gamma(u;c')-\gamma(u;c)\} du
\\=&I_1+I_2\quad \text{say}.
\end{align*}
For $I_1$, using integration by parts, we have
\begin{align*}
I_1=&-\int \phi(x-u)\gamma(u;c')\nabla\log \gamma(u;c')\{\nabla \log\gamma(u;c')-\nabla \log\gamma(u;c)\} du
\\&-\int \phi(x-u)\gamma(u;c')\{\nabla^2 \log\gamma(u;c')-\nabla^2 \log\gamma(u;c)\} du
\\=&I_{11}+I_{12}\quad \text{say}.
\end{align*}
Here $I_{11}$ can be bounded in a similar way as in Step 2. Under
(\ref{eq-ass}), it is straightforward to see that
$|I_{12}/m(x;\theta')|\lesssim |c-c'|/w'$. Notice that in the case
of double exponential distribution, we have
\begin{align*}
|I_{12}|\lesssim &|\phi(x-c')\gamma_0(0)-\phi(x-c)\gamma_0(c-c')|
\\ \lesssim &\phi(x-c')|\gamma_0(c-c')-\gamma_0(0)|+\gamma_0(c-c')|\phi(x-c)-\phi(x-c')|,
\end{align*}
which implies that $|I_{12}/m(x;\theta')|\lesssim
|c-c'|e^{\Lambda|c-c'|}/w'$.

On the other hand, we have
\begin{align*}
I_2=&-\int \phi(x-u)\nabla^2 \log\gamma(u;c)\{\gamma(u;c')-\gamma(u;c)\} du
\\&-\int \phi(x-u)\nabla\log\gamma(u;c)\{\nabla\gamma(u;c')-\nabla\gamma(u;c)\} du,
\end{align*}
which can be handled in a similar way as in Step 2. Combining the arguments, we can show that
\begin{align*}
\left|\frac{\nabla^2 m(x;\theta)-\nabla^2 m(x;\theta')}{m(x;\theta')}\right|\lesssim |w-w'|e^{\Lambda|c'|}/w'+|c-c'|e^{\Lambda|c-c'|}/w'.
\end{align*}

\noindent \emph{ Step 4:} Combining Steps 1-3 and using the arguments in Case 1, we can show that
\begin{align*}
&p^{-1}|\{\hat{R}(w,c)-E\hat{R}(w,c)\}-\{\hat{R}(w',c')-E\hat{R}(w',c')\}|
\\ \lesssim & \max(\lambda_0^2,c_0)\lambda_0\left(|w-w'|e^{\Lambda
c_0}+|c-c'|\right).
\end{align*}

\noindent \emph{ Step 5:} The rest of the proof is similar to those in Case 1. Set
$w_j=\delta j$ and $c_i=\delta' i$ for $w_j\in [1/\lambda_0,1]$ and
$c_i\in [-c_0,c_0]$. Choose
$\max(\lambda_0^2,c_0)\lambda_0\left(\delta e^{\Lambda
c_0}+\delta'\right)=o(1/\sqrt{p})$. Then we have
\begin{align*}
\tilde{A}=\left\{\max_{w\in [1/\lambda_0,1],|c|\leq c_0}p^{-1}|\hat{R}(w,c)-E\hat{R}(w,c)|\geq 2\epsilon/\sqrt{p}\right\}\subseteq \tilde{D},
\end{align*}
where
$$\tilde{D}=\left\{\max_{i,j}p^{-1}|\hat{R}(w_j,c_i)-E\hat{R}(w_j,c_i)|\geq \epsilon/\sqrt{p}\right\}.$$
Again using the union bound and the Hoeffding's inequality, we have
\begin{align*}
P(\tilde{A})\leq \frac{16(\lambda_0-1)c_0}{\lambda_0\delta\delta'}\exp\left[-\frac{2\epsilon^2}{\{c_1 +
c_2(|c|+\log(\lambda_0))\}^2}\right].
\end{align*}
Picking $\epsilon^2=s^2\log(p)\{c_1 +
c_2(|c|+\log(\lambda_0))\}^2/2,$ we get
\begin{align*}
P\left(\max_{w\in [1/\lambda_0,1],|c|\leq c_0}\frac{1}{\sqrt{p\log
(p)}\{c_1 +
c_2(|c|+\log(\lambda_0))\}}|\hat{R}(w,c)-E\hat{R}(w,c)|\geq
\sqrt{2}s\right)\leq
\frac{16(\lambda_0-1)c_0}{\lambda_0\delta\delta'}p^{-s^2}.
\end{align*}
For $\lambda_0=a_1p^{a_2}$, $c_0=a_3\log(p)$,
$\delta=p^{-a_3\Lambda-1/2-3a_2-\varepsilon}$, $\delta'=p^{-1/2-3a_2-\varepsilon}$ and large enough $s$ where $\varepsilon>0$, we have $\max(\lambda_0^2,c_0)\lambda_0\left(\delta e^{\Lambda
c_0}+\delta'\right)=o(1/\sqrt{p})$ and
\begin{align*}
\max_{w\in [1/\lambda_0,1],|c|\leq
c_0}p^{-1}|\hat{R}(w,c)-E\hat{R}(w,c)|=O_p\left(\frac{(\log(p))^{3/2}}{\sqrt{p}}\right).
\end{align*}

\subsection{EM+PAV algorithm for MMLE}\label{sec:empav}

\begin{algorithm}[H]\label{alg3}\small
\caption{}
0. Input $d$ and the initial values $(w^{(0)}_0,w^{(0)}_1,c^{(0)}_1,\dots,w^{(0)}_{d},c^{(0)}_{d})$ and $(b^{(0)}_{1i},\dots,b^{(0)}_{di})$ for $1\leq i\leq p$.\\
1. \textbf{E-step:} Given  $(w_0,w_1,c_1,\dots,w_{d},c_{d})$ and
$(b_{1i},\dots,b_{di})$ for $1\leq i\leq p$, let
  $$Q_{0i}=\frac{(1-w_0)\phi(Y_i)}{(1-w_0)\phi(Y_i)+\sum^{d}_{j=1}w_jg(X_i;\tau_{ji}^{-1/2},c_j/\sigma_i)},$$
  and
  $$Q_{ki}=\frac{w_kg(Y_i;\tau_{ki}^{-1/2},c_k)}{(1-w_0)\phi(X_i)+\sum^{d}_{j=1}w_jg(Y_i;\tau_{ji}^{-1/2},c_j/\sigma_i)},$$
  for $1\leq k\leq d,$ where $\tau_{ki}=1/(\sigma_i^2b_k^2)$.\\
2. \textbf{M-step:} For fixed $(c_1,\dots,c_{d})$, solve the weighted isotonic regression,
\begin{align}\label{opt-mon3}
(\tilde{\tau}_{k1},\dots,\tilde{\tau}_{kp})=\arg\min\sum^{p}_{i=1}Q_{ki}\left\{(Y_i-c_k/\sigma_i)^2-1-\tau_{ki}\right\}^2
\quad \text{subject to}\quad 0\leq \tau_{ki}\leq \tau_{kj} \quad \text{if}
\quad \sigma_i\geq \sigma_j.
\end{align}
Let $\hat{\tau}_{ki}=\max\{\tilde{\tau}_{ki},0\}$ for $1\leq i\leq p.$
For fixed $(\tau_{k1},\dots,\tau_{kp})$, let
\begin{align}\label{opt-cw3}
\hat{c}_k=\frac{\sum_{i=1}^{p}Q_{ki}Y_i/\{\sigma_i(1+\tau_{ki})\}}{\sum_{i=1}^{p}Q_{ki}/\{\sigma_i^2(1+\tau_{ki})\}}\quad
\text{and}\quad \hat{w}_k=\frac{1}{p}\sum^{p}_{i=1}Q_{ki},
\end{align}
with $1\leq k\leq d.$ Iterate between (\ref{opt-mon3}) and
(\ref{opt-cw3}) until
convergence.\\
3. Repeat the above E-step and M-step until the algorithm converges.
\end{algorithm}

\bibliographystyle{apalike}
\bibliography{refs}

\end{document}